\documentclass[english,12pt]{article}
\usepackage[latin1]{inputenc}
\usepackage{babel}
\usepackage{amsxtra}
\usepackage{amsthm}
\usepackage{amsfonts}
\usepackage{amsmath}
\usepackage{amssymb}
\usepackage{mathrsfs}
\usepackage{stmaryrd}
\usepackage{dsfont}
\usepackage{bbm}
\usepackage{hyperref}

\def\one{\mathbbm{1}}
\newtheorem{thm}{Theorem}[section]
\newtheorem{prop}[thm]{Proposition}
\newtheorem{cor}[thm]{Corollary}
\newtheorem{lemma}[thm]{Lemma}

\theoremstyle{definition}
\newtheorem{defi}[thm]{Definition}
\newtheorem{remark}[thm]{Remark}
\newtheorem{example}[thm]{Example}

\newtheorem*{ak}{Acknowledgement}

\newcommand{\bt}{\begin{thm}}
\newcommand{\et}{\end{thm}}
\newcommand{\br}{\begin{remarks}}
\newcommand{\er}{\end{remarks}}
\newcommand{\bl}{\begin{lemma}}
\newcommand{\el}{\end{lemma}}
\newcommand{\bp}{\begin{proof}}
\newcommand{\ep}{\end{proof}}
\newcommand{\bal}{\begin{align*}}
\newcommand{\eal}{\end{align*}}
\newcommand{\bi}{\begin{itemize}}
\newcommand{\be}{\begin{equation}}
\newcommand{\ee}{\end{equation}}
\newcommand{\bea}{\begin{eqnarray}}
\newcommand{\eea}{\end{eqnarray}}
\newcommand{\ei}{\end{itemize}}
\newcommand{\sint}{\stackrel{\mbox{\tiny$\mskip-6mu\bullet\mskip-6mu$}}{}}
\newcommand{\ssint}{\stackrel{\mbox{\tiny$\scriptscriptstyle\bullet$}}{}}

\DeclareMathOperator{\essinf}{ess\ inf}
\DeclareMathOperator{\sign}{sign}
\DeclareMathOperator{\argmin}{argmin\, }
\DeclareMathOperator{\Var}{Var}

\newcommand{\Inf}{\inf\limits}
\newcommand{\Lim}{\lim\limits}
\newcommand{\bJ}{\bar{J}}
\newcommand{\bL}{\bar{L}}

\newcommand{\R}{\mathbb{R}}
\newcommand{\RR}{\mathbb{R}}
\newcommand{\N}{\mathbb{N}}

\newcommand{\F}{\mathcal{F}}

\newcommand{\cC}{\mathcal{C}}
\newcommand{\cL}{\mathcal{L}}
\newcommand{\Om}{\Omega}
\newcommand{\om}{\omega}
\newcommand{\Ombar}{\overline{\Om}}

\newcommand{\LiiP}{L^2(P)}
\newcommand{\LiiM}{\cL^2(M)}
\newcommand{\LiiA}{\cL^2(A)}

\newcommand{\Lzloc}{L^2_{\rm loc}}
\newcommand{\HzP}{\mathcal{H}^2(P)}
\newcommand{\HzlocP}{\mathcal{H}^2_{\rm loc}(P)}

\newcommand{\MznlocP}{\mathcal{M}^2_{0,\rm loc}(P)}
\newcommand{\MnlocP}{\mathcal{M}_{0,\rm loc}(P)}

\newcommand{\cS}{\mathcal{S}}

\newcommand{\cP}{\mathcal{P}}

\def\Whell{\widehat{W^\ell}}
\def\Whellpm{\widehat{W^{\ell^\pm}}}
\def\Whellmp{\widehat{W^{\ell^\mp}}}
\def\WhL{\widehat{W^L}}
\def\ZN{Z^N}
\def\ZNTm{\ZN_{T_m}}

\newcommand{\fC}{\mathfrak{C}}
\newcommand{\fK}{\mathfrak{K}}
\newcommand{\fg}{\mathfrak{g}}

\newcommand{\fh}{\mathfrak{h}}

\newcommand{\TB}{\overline{\Theta}}
\newcommand{\TBC}{\overline{\Theta(C)}}
\newcommand{\TBK}{\overline{\Theta(K)}}
\newcommand{\vt}{\vartheta}
\newcommand{\tvt}{\widetilde\vartheta}
\newcommand{\vp}{\varphi}
\newcommand{\vr}{\varrho}
\newcommand{\ve}{\varepsilon}
\newcommand{\tpsi}{\widetilde\psi}

\newcommand{\TS}{\Theta}
\newcommand{\TSC}{\Theta(C)}
\newcommand{\TSK}{\Theta(K)}

\newcommand{\omt}{(\om,t)}
\newcommand{\OmT}{\Om\times[0,T]}
\newcommand{\tvp}{\widetilde\vp}
\newcommand{\PtB}{P\otimes B}
\newcommand{\la}{\langle}
\newcommand{\ra}{\rangle}
\newcommand{\E}{\mathcal{E}}
\newcommand{\cE}{\mathcal{E}}
\newcommand{\T}{\top}
\def\marteq{\mathrel{\mathop{=}\limits^{\rm mart}}}

\numberwithin{equation}{section}
\begin{document}
%\nocite{*}

\title{\vspace{-2\baselineskip}
Cone-Constrained Continuous-Time\\ Markowitz Problems}

\author{
Christoph Czichowsky\\
Faculty of Mathematics, University of Vienna\\
Nordbergstrasse 15, A--1090 Vienna, Austria \\
{\tt christoph.czichowsky@univie.ac.at}\\
\\
Martin Schweizer\\
Department of Mathematics, ETH Zurich\\
R\"amistrasse 101, CH--8092 Zurich, Switzerland \\
{\sl and}\\
Swiss Finance Institute\\
Walchestrasse 9, CH--8006 Zurich, Switzerland \\
{\tt martin.schweizer@math.ethz.ch}
\\
\vspace{-0.2cm}\\
This version: February 24, 2012.\footnote{To appear in \emph{Annals of Applied Probability.}}
}
\date{}
\maketitle
\vspace{-0.4cm}
\begin{abstract}
\noindent
The Markowitz problem consists of finding in a financial market a self-financing trading strategy whose final wealth has maximal mean and minimal variance. We study this in continuous time in a general semimartingale model and under cone constraints: Trading strategies must take values in a (possibly random and time-dependent) closed cone. We first prove \emph{existence} of a solution for convex constraints by showing that the space of constrained terminal gains, which is a space of stochastic integrals, is closed in $L^2$. Then we use stochastic control methods to describe the \emph{local structure} of the optimal strategy, as follows. The value process of a naturally associated constrained linear-quadratic optimal control problem is decomposed into a sum with two opportunity processes $L^\pm$ appearing as coefficients. The martingale optimality principle translates into a drift condition for the semimartingale characteristics of $L^\pm$ or equivalently into a coupled system of backward stochastic differential equations for $L^\pm$. We show how this can be used to both characterise and construct optimal strategies. Our results explain and generalise all the results available in the literature so far. Moreover, we even obtain new sharp results in the unconstrained case.
\end{abstract}



\newpage
\section{Introduction}
Mean-variance portfolio selection is a classical problem in finance. It consists of finding in a financial market a self-financing trading strategy whose final wealth has maximal mean and minimal variance. It is often called the \emph{Markowitz problem} after its inventor Harry Markowitz who proposed it in a one-period setting as a formulation for portfolio optimisation; see \cite{M52} and \cite{M59}. We study this problem here in continuous time in a general semimartingale model and under \emph{cone constraints}, meaning that each allowed trading strategy is restricted to always lie in a closed cone which might depend on the state and time in a predictable way. For applications in the management of pension funds and insurance companies, the inclusion of such constraints into the setup is very useful as they allow to model regulatory restrictions, like for example no shortselling.

As in the unconstrained case, the solution to the Markowitz problem can be obtained by solving the particular mean-variance hedging problem of approximating in $L^2$ a constant payoff by the terminal gains of a self-financing trading strategy. To get existence of a solution to the latter problem, we show first that the space $G_T(\fC)$ of constrained terminal gains is closed in $L^2$; this is sufficient if the constraints, and hence $G_T(\fC)$, are in addition convex. Our approach here combines the space of ($L^2$-)admissible trading strategies of \v Cern\'y and Kallsen \cite{CK07} with $\E$-martingales, a generalisation of martingales introduced by Choulli, Krawczyk and Stricker \cite{CKS98}. The latter notion comes up naturally in quadratic optimisation problems in mathematical finance due to the negative ``marginal utility'' of the square function. The \emph{closedness} result and hence the \emph{existence} of optimal strategies for the constrained Markowitz problem constitute a first major contribution, especially in view of the generality of our setting. In particular, this allows us to obtain in Theorem \ref{T6.2} a new sharp result for the unconstrained case.

Our main focus and achievement, however, is the subsequent \emph{structural description} of the optimal strategy by its local properties. This is made possible by treating the approximation in $L^2$ as a problem in stochastic optimal control and systematically using ideas and results from there. By exploiting the quadratic and conic structure of our task, we first obtain a decomposition of its value process $J(x,\vt)$ into a sum involving two auxiliary coefficient processes. This is similar to the results by \v Cern\'y and Kallsen \cite{CK07} in the unconstrained case, but now requires two \emph{opportunity processes} $L^\pm$, due to the constraints. An analogous opportunity process also plays a central role in the analysis by Nutz \cite{N09} of power utility maximisation, and some of the ideas and techniques are similar. Using the martingale optimality principle for $J(x,\vt)$ next allows us to describe first the drift of $L^\pm$ and from there the optimal strategy locally in feedback form via the pointwise minimisers of two predictable functions ${\fg}^\pm$; these are given in terms of the joint differential semimartingale characteristics of the opportunity processes $L^\pm$ and the price process $S$. The drift equations can also be rewritten as a system of coupled backward stochastic differential equations (BSDEs) for $L^\pm$, and we show that the opportunity processes are the maximal solutions of this system. This is motivated by a similar result in \cite{N09}. Conversely, we also prove verification results saying that if we have minimisers of ${\fg}^\pm$ (or a solution to the BSDE system), then we can construct from there an optimal strategy. This explains and generalises all results so far in the literature on the Markowitz problem under cone constraints; see \cite{LLZ02}, \cite{HZ04}, \cite{LH07} and \cite{JZ07}. 

The generality of our framework allows us to capture a new behaviour of the optimal strategy: It jumps from the minimiser of one predictable function to that of a second one, whenever the optimal wealth process of the approximation problem changes sign. Because this phenomenon is due to jumps in the price process $S$ of the underlying assets, it could not be observed in earlier work since the Markowitz problem under constraints has so far only been studied in (continuous) It\^o process models. Not surprisingly, the presence of jumps and the resulting nontrivial coupling of the BSDEs make the situation more involved; we explain in Section \ref{sec:lit} how things quickly simplify if $S$ is continuous. The usefulness of our general results can also be illustrated by applying them to L\'evy processes. Here the two random equations for the joint differential characteristics of $L^\pm$ and $S$  reduce to two coupled ordinary differential equations. These allow us to describe the solution explicitly, and it turns out that its behaviour is quite different than in the unconstrained case; the details and examples illustrating the various effects have been worked out and will be presented elsewhere.

The paper is organised as follows. Section \ref{sec:pf} gives a precise formulation of the problem, recalls basic results on predictable correspondences and proves the closedness in $L^2$ of the space of constrained terminal gains. In Section \ref{sec:dp}, we use dynamic programming arguments to establish the general structure of the value process $J(x,\vt)$ in terms of the opportunity processes $L^\pm$. Section \ref{sec:ld} exploits this via the martingale optimality principle to derive the local description of the optimal strategy and the characterisation of the opportunity processes via coupled BSDEs. Section \ref{sec:pr} contains the more computational parts of the proofs from Section \ref{sec:ld}, and Section \ref{sec:lit} concludes with a comparison to related work.
%
%%%%%%%%%%%%%%%%%%%%%%%%%%%%%%%%%%%%%%%%%%%%%%%%%%%%%%%%%%
%
\section{Formulation of the problem and preliminaries}\label{sec:pf}
Let $(\Omega,\F,P)$ be a probability space with a filtration $\mathbb{F}=(\F_t)_{0\leq t\leq T}$ satisfying the usual conditions of completeness and right-continuity, where $T>0$ is a fixed and finite time horizon. We can and do choose for every local $P$-martingale a right-continuous version with left limits (RCLL for short). All unexplained notation concerning stochastic integration can be found in the books of Jacod and Shiryaev \cite{JS} and Protter \cite{P04}. For local martingales, we use the definition in \cite{P04}.

We consider a \emph{financial market} consisting of one riskless asset, whose (discounted) price is $1$, and $d$ risky assets described by an $\R^d$-valued RCLL semimartingale $S=(S_t)_{0\leq t\leq T}$. We suppose that $S$ is \emph{locally square-integrable}, $S\in \HzlocP$, in the sense that $S$ is special with canonical decomposition $S=S_0+M+A$, where $M$ is an $\R^d$-valued locally square-integrable local martingale null at zero, $M\in \MznlocP$, and $A$ is an $\R^d$-valued predictable RCLL process of finite variation and null at zero. Using semimartingale characteristics, we write $\la M\ra=\widetilde c^M\sint B$ and $A=b^S\sint B$, where all processes are predictable, $B$ is RCLL and strictly increasing and null at $0$, and $\widetilde c^M$ is $d\times d$-matrix-valued. For details, see Section II.2 in \cite{JS} or Section \ref{sec:ld} below. On the product space $\Ombar:=\OmT$ with the predictable $\sigma$-field $\cP$, define $P_B:=\PtB$. As \emph{trading strategies} available for investment, we consider a set $\fC$ of $S$-integrable, $\R^d$-valued, predictable processes; this will be specified more precisely later. We call $\fC$ \emph{unconstrained} if $\fC$ is a linear subspace and \emph{constrained} otherwise. By trading with a strategy $\vt\in\fC$ up to time $t\in[0,T]$ in a self-financing way, an investor with initial capital $x\in\R$ can generate the \emph{wealth} 
$$
\textstyle
V_t(x,\vt):=x+\int_0^t\vt_u \,dS_u=:x+\vt\sint S_t.
$$
In this paper, we understand \emph{mean-variance portfolio selection} as in the usual \emph{Marko\-witz problem}, i.e.~as the \emph{static} optimisation problem of finding a (dynamic) self-financing trading strategy whose final wealth has maximal mean and minimal variance. This is static in the sense that we only consider the optimisation at the initial time $0$ without looking at intermediate conditional versions. Mathematically, this can be formulated as
\be
\text{maximise $E[V_T(x,\vartheta)]-\frac{\gamma}{2}\Var[V_T(x,\vartheta)]$ over all $\vt\in\fC$,}\label{MVPS}
\ee
where the parameter $\gamma>0$ describes the \emph{risk aversion} of the investor. The most common alternative formulation is to 
\begin{align}
&\text{minimise $\Var[V_T(x,\vartheta)]=E\big[|V_T(x,\vartheta)|^2\big]-m^2$}\nonumber\\
&\text{subject to $E[V_T(x,\vartheta)]=m>x$ and $\vartheta\in\fC$.}\label{CMP}
\end{align}
If $\fC=\fK$ is a \emph{cone}, we obtain from the purely geometric structure of the optimisation problems the following global description of the solutions to \eqref{MVPS} and \eqref{CMP}.
\bl\label{lgds}
If $\fC=\fK$ is a \emph{cone} and if we have $\tvp\sint S_T\not\equiv1$ and $E[\tvp\sint S_T]>0$, then the solutions to \eqref{MVPS} and \eqref{CMP} are given by
\be
\tvt=\frac{1}{\gamma}\frac{1}{E[1-\widetilde\vp\sint S_T]}\widetilde\vp
\qquad\text{and}\qquad
\tvt^{(m,x)}=\frac{m-x}{E[\widetilde\vp\sint S_T]}\widetilde\vp,\label{gds}
\ee
respectively, where $\widetilde\vp$ is the solution to 
\be
\text{minimise $E\big[|V_T(-1,\vt)|^2\big]=E\big[ |1-\vt\sint S_T|^2\big]$ over all $\vt\in\fC$.}\label{ap-1}
\ee
\el
\bp
This follows from the arguments in the proof of Proposition 3.1 and Theorem 4.2 in \cite{SW06} which are derived in an abstract $L^2$-setting by Hilbert space arguments. Note that the convexity assumed in \cite{SW06} is not necessary for the equations \eqref{gds} to hold; it is used in \cite{SW06} only for the existence of a solution to \eqref{ap-1}, which we do not assert here.
\ep
If $\fC$ is a \emph{convex set}, but not necessarily a cone, one can under suitable feasibility conditions still establish the \emph{existence} of a solution to \eqref{MVPS} and \eqref{CMP} by using Lagrange multipliers; see \cite{LH07} and \cite{D08}. However, these solutions admit less \emph{structure} so that their dynamic behaviour over time cannot be described very explicitly. We therefore concentrate from Section \ref{sec:dp} onwards on constraints which are given by \emph{cones}. Before that, however, we want to prove existence of an optimal strategy in a continuous-time setting.

We first observe that despite its simplicity, Lemma \ref{lgds} is very useful as it relates the solution to the Markowitz problems \eqref{MVPS} and \eqref{CMP} to the solution of a \emph{constrained mean-variance hedging problem}, namely minimising the mean-squared hedging error between a given payoff $H\in\LiiP$ and a constrained self-financing trading strategy, i.e.~to
\begin{equation}
\text{minimise $E\left[|V_T(x,\vt)-H|^2\right]=E\left[|x+\vt\sint S_T-H|^2\right]$ over all $\vt\in\fC$.}\label{MVH}
\end{equation}
Indeed, \eqref{ap-1} corresponds to the very particular version of this problem with $H\equiv0$ and $x=-1$, or $H\equiv1$ and $x=0$. Since \eqref{MVH} is an approximation problem in the Hilbert space $ \LiiP$, it admits a solution for arbitrary $H\in  \LiiP$ if the space 
$$G_T(\fC)=\{\vt\sint S_T \,|\, \vt\in\fC\}$$
of terminal constrained gains is convex and \emph{closed} in $ \LiiP$. Such constrained mean-variance hedging problems in a general semimartingale framework have been studied in \cite{CS10a}. As explained there, one can formulate constraints on trading strategies and then adapt closedness results from the unconstrained case to obtain closedness under constraints as well. This needs a suitable choice of strategies and constraints which we now introduce.

Conceptually, our choice of strategy space can be traced back to \v Cern\'y and Kallsen \cite{CK07}. They start with simple integrands of the form $\vartheta=\sum_{i=1}^{m-1}\xi_i I_{\rrbracket \sigma_i,\sigma_{i+1}\rrbracket }$ with stopping times \hbox{$0\leq\sigma_1\leq\cdots\leq\sigma_m\le \tau_n \leq T$} for some $n\in\N$ and bounded $\R^d$-valued $\F_{\sigma_i}$-measurable random variables $\xi_i$ for $i=1,\ldots,m-1$, where $(\tau_n)$ is a localising sequence of stopping times with $S^{\tau_n}\in \HzP$. Their ($L^2$-)admissible strategies are then those integrands $\vt\in \cL(S)$ for which there exists a sequence $(\vt^n)_{n\in\N}$ of simple integrands such that
\begin{itemize}
\item[\bf{1)}] $\vt^n\sint S_T\overset{ \LiiP}{\longrightarrow} \vt\sint S_T$.
\item[\bf{2)}] $\vt^n\sint S_t\overset{P}{\longrightarrow} \vt\sint S_t$ for all $t\in[0,T]$.
\end{itemize}
A discussion why such a class of strategies is economically reasonable and mathematically useful can be found in \cite{CK07}. For our purposes, we need to modify that definition a little.

Instead of simple strategies, another natural space of strategies coming from the construction of the stochastic integral is $\Theta:=\TS_S:=\LiiM\cap \LiiA$ with
\begin{align*}
\LiiM
&:=
\textstyle
\big\{\vartheta\in\cL^0(\Ombar,\cP;\R^d)\, \big|\, \|\vartheta\|_{\LiiM}:=\big(E\big[\int_0^T\vartheta_s
^\T d\la M\ra_s \, \vartheta_s\big]\big)^{\frac{1}{2}}<\infty\big\},
\\
\LiiA
&:=
\textstyle
\big\{\vartheta\in\cL^0(\Ombar,\cP;\R^d)\, \big|\,  \|\vartheta\|_{\LiiA}:=\big(E\big[\big(\int_0^T|\vartheta_s^\T dA_s|\big)^2\big]\big)^{\frac{1}{2}}<\infty\big\}.
\end{align*}
Next, the trading constraints we consider are formulated via predictable correspondences. 
\begin{defi}
A \emph{correspondence} is a mapping $C:\Ombar\to 2^{\RR^d}$. We call a correspondence $C$ \emph{predictable} if $C^{-1}(F):=\{\omt \,|\, C\omt\cap F\ne\emptyset\}$ is a predictable set for all closed sets $F\subseteq\RR^d$. The \emph{domain} of a correspondence $C$ is $\mathrm{dom}(C):=\{\omt\, |\, C\omt\ne\emptyset\}$. A \emph{(predictable) selector} of a (predictable) correspondence $C$ is a (predictable) process $\psi$ with $\psi\omt\in C\omt$ for all $\omt\in\mathrm{dom}(C)$.
\end{defi}
For a correspondence $C:\Ombar\to 2^{\RR^d}\setminus\{\emptyset\}$, the sets of \emph{$C$-valued} or \emph{$C$-constrained integrands} and of \emph{square-integrable $C$-constrained trading strategies} are given by
\begin{align*}
\cC
&:=
\cC^{S}:=\{\vt\in\cL(S) \,|\, \vt\omt\in C\omt \text{ for all $\omt\in\Ombar$}\},
\\
\Theta(C)
&:=
\Theta\cap\cC=\{\vt\in\Theta \,|\, \vt\omt\in C\omt \text{ for all $\omt\in\Ombar$}\}.
\end{align*}
\begin{defi}
A trading strategy $\vt\in \cC$ is called \emph{$C$-admissible (in $ \LiiP$)} if there exists a sequence $(\vt^n)_{n\in\N}$ in $\Theta(C)$, called \emph{approximating sequence for $\vt$}, such that
\begin{itemize}
\item[\bf{1)}] $\vt^n\sint S_T\overset{ \LiiP}{\longrightarrow} \vt\sint S_T$.
\item[\bf{2)}] $\vt^n\sint S_\tau\overset{P}{\longrightarrow} \vt\sint S_\tau$ for all stopping times $\tau$.
\end{itemize}
The set of all $C$-admissible trading strategies is called $\overline{\Theta(C)}$, and we set $\TB := \overline{\Theta(\R^d)}$.
\end{defi}
In comparison to \v Cern\'y and Kallsen \cite{CK07}, there are two differences. Instead of using simple strategies for the approximation, we use strategies from $\Theta(C)$; the reason is that it can easily happen with time-dependent constraints that no simple strategy satisfies them. (The constraints can also be so bad that no strategy in $\Theta$ satisfies them either; but such situations are almost pathological.) The second difference is that we stipulate 2) for all stopping times $\tau$ and not only for deterministic times $t$; this is needed for dynamic programming arguments, as explained at the end of this section.

Before addressing the issue of closedness of $G_T(\TBC)$ in $ \LiiP$, we recall some results on predictable correspondences, used later to ensure the existence of predictable selectors.
\begin{prop}[Castaing]\label{Castaing}
For a correspondence $C:\Ombar\to2^{\RR^d}$ with closed values, the following are equivalent:
\bi
\item[\bf 1)] $C$ is predictable.
\item[\bf 2)] $\mathrm{dom}(C)$ is predictable and there exists a \emph{Castaing representation} of $C$, i.e.~a sequence $(\psi^n)$ of predictable selectors of $C$ such that
$$C\omt=\overline{\{\psi^1\omt,\psi^2\omt,\ldots\}}
\quad\text{for each $\omt\in\mathrm{dom}(C)$.}$$
\ei
In particular, every closed-valued predictable $C$ admits a predictable selector $\psi$.
\end{prop}
\begin{proof}
See Corollary 18.14 in \cite{AB} or Theorem 1B in \cite{Roc75}.
\end{proof}
%%\marginpar{propCara}
\begin{prop}\label{propCara}
%%\marginpar{Check: non-empty?}
Let $C:\Ombar\to2^{\RR^d}$ be a predictable correspondence with closed values and $f:\Ombar\times\R^m\to\R^d$ and $g:\Ombar\times\R^d\to\R^m$ \emph{Carath\'eodory functions}, which means that $f(\omega,t,y)$ and $g(\omega,t,x)$ are predictable with respect to $\omt$ and continuous in $y$ and $x$. Then the mappings $C'$ and $C''$ given by $C'\omt=\{y\in\R^m\, |\, f(\omega,t,y)\in C\omt\}$ and $C''\omt=\overline{\{g(\omega,t,x)\, |\, x\in C\omt\}}$ are predictable correspondences (from $\Ombar$ to $2^{\R^m}$) with closed values. 
\end{prop}
\begin{proof}
See Corollaries 1P and 1Q in \cite{Roc75}.
\end{proof}
\begin{prop}
% \label{propcount}
Let $C^n:\Ombar\to2^{\RR^d}$ for each $n\in\N$ be a predictable correspondence with closed values and define the correspondences $C'$ and $C''$ by $C'\omt=\underset{n\in\N}{\bigcap}C^n\omt$ and $C''\omt=\underset{n\in\N}{\bigcup}C^n\omt$. Then $C'$ and $C''$ are predictable and $C'$ is closed-valued.
\end{prop}
\begin{proof}
See Theorem 1M in \cite{Roc75} and Lemma 18.4 in \cite{AB}.
\end{proof}

Now we aim to prove closedness in $ \LiiP$ of the space of constrained terminal gains. Because we are interested in solving \eqref{ap-1} in maximal generality, we combine ideas and concepts from \cite {CK07} and \cite{CKS98}. Like \v Cern\'y and Kallsen in \cite{CK07}, we use the (modified) space $G_T(\TBC)$ of ($L^2$-)admissible trading strategies, but we drop the assumption from \cite{CK07} that there exists an equivalent local martingale measure (ELMM) $Q$ for $S$ with $\frac{dQ}{dP}\in  \LiiP$. (To illustrate why this is useful, consider the simple case where $S$ is a Poisson process. Then one can compute straightforwardly that the solution to \eqref{ap-1} is given by $\tvp= \one_{\llbracket 0,\tau\rrbracket}$, where $\tau=\Inf\{t>0 \,|\, \Delta S_t=1\}\wedge T$. However, there exists no E(L)MM because each integrand $\vt\equiv c>0$ is an arbitrage opportunity.) Like Choulli, Krawczyk and Stricker in \cite{CKS98}, we impose instead of the existence of an ELMM $Q$ the more general absence-of-arbitrage condition that $S$ is an $\E$-local martingale; but unlike \cite{CKS98}, we do not require a reverse H\"older inequality. 

Let us first recall the notion of an $\E$-martingale. For a semimartingale $Y$, we denote its stochastic exponential by $\E(Y)$. {\bf Throughout this paper, we let $N$ stand for a local $P$-martingale null at zero and $\ZN$ for a strictly positive adapted RCLL process.} We shall see below how $N$ and $\ZN$ are related. For any stopping time $\tau$, we denote the process $Y$ stopped at $\tau$ by $Y^\tau$ and the process $Y$ started at $\tau$ by ${}^\tau Y := Y-Y^\tau$; but we set ${}^\tau\E(N) := \E(N-N^\tau)$. So for stochastic exponentials, ${}^\tau\E(N)$ denotes a multiplicative rather than an additive restarting. Since $N$ is RCLL, it has at most a finite number of jumps with $\Delta N=-1$, and so there is $P$-a.s.~at most a finite number of times, not depending on $\tau$, where the ${}^\tau\E(N)$ can jump to zero; this follows from the representation of the stochastic exponential in Theorem II.37 in \cite{P04}. Thus the stopping times defined by $T_0=0$ and $T_{m+1}=\Inf\{t>T_m\, |\, {}^{T_m}\E(N)_t=0\}\wedge T$ increase stationarily to $T$.
\begin{defi}
An adapted RCLL process $Y$ is an \emph{$\E$-local martingale} if the product of ${}^{T_m}Y$ and ${}^{T_m}\E(N)$ is a local $P$-martingale for any $m\in\N$. It is an \emph{$(\E, \ZN)$-martingale} if for any $m\in\N$, we have
$
E\big[|Y_{T_m} \ZNTm  {}^{T_m}\E(N)_{T_{m+1}}|\big]<\infty
$
and the product of ${}^{T_m}Y$ and $\ZNTm  {}^{T_m}\E(N)$ is a (true) $P$-martingale.
\end{defi}
In comparison to Definition 3.11 in \cite{CKS98}, we have generalised the definition of $\E$-mar\-tin\-gales to $(\E,\ZN)$-martingales by introducing the process $\ZN$. This is needed for a clean formulation of our results, but it also makes intuitive sense. Suppose $Q$ is an equivalent martingale measure for $Y$ and write its density process with respect to $P$ as $Z^Q = Z^Q_0 \E(N^Q)$. By the Bayes rule, the product $Y Z^Q$ is then a $P$-martingale and so is ${}^0 Y Z^Q = (Y-Y_0) Z^Q$. One consequence is that the product of ${}^0 Y$ and $\E(N^Q)$ is a local $P$-martingale so that $Y$ is an $\E(N^Q)$-local martingale. (Of course, $Z^Q>0$ implies that $T_m \equiv T$ for $m\ge1$.) We also have that ${}^0 Y Z^Q_0 \E(N^Q)$ is a true $P$-martingale so that $Y$ is an $(\E(N^Q),Z^Q)$-martingale. But unless we know more about $Z^Q_0$, we cannot assert that the product ${}^0 Y \E(N^Q)$ is a true $P$-martingale  (since it need not be $P$-integrable); so $Y$ is not an $\E(N^Q)$-martingale in the sense of \cite{CKS98}. Hence we see that in the abstract definition, $\ZNTm$ plays a similar role at time $T_m$ as the density $Z^Q_0$ of $Q$ at time 0, and its main role is to ensure integrability properties. (This is not needed in \cite{CKS98} because the authors there work with $Y = \vt\sint S \in \HzP$ and assume that $\E(N)$ satisfies the reverse H\"older inequality $R_2(P)$. In our notation, this allows to take $\ZN\equiv 1$.)
\begin{remark}
\label{Nrem}
If $N$ is as above a local martingale, then 
$J^m := \one_{\rrbracket T_m,T\rrbracket}\sint\E( \one_{\rrbracket T_m,T\rrbracket}\sint N)$
is for each $m$ also a local martingale; if $N$ is in addition locally square-integrable, then so is $J^m$; and both statements still hold if we multiply $J^m$ by a strictly positive $\F_{T_m}$-meas\-ur\-able random variable. There is no problem with adaptedness since $J^m=0$ on $\rrbracket T_m,T\rrbracket$.

Conversely -- and this will be used later -- suppose $N$ is a semimartingale. If $J^m$ is for each $m$ a local martingale, then writing
$J^m = (\E(\one_{\rrbracket T_m,T\rrbracket}\sint N)_- \one_{\rrbracket T_m,T\rrbracket} )\sint N$
and observing that
$\E( \one_{\rrbracket T_m,T\rrbracket}\sint N)_-\ne0$ on $\rrbracket T_m,T_{m+1}\rrbracket$
by the definition of $T_m$ shows that $\one_{\rrbracket T_m,T_{m+1}\rrbracket}\sint N$ is a local martingale for each $m$, and then so is $N$. Again this still holds if we replace $J^m$ by $\beta_m J^m$ for an $\F_{T_m}$-measurable $\beta_m>0$, and again local square-integrability transfers, from $J^m$ (or $\beta_m J^m$) to $N$.
\end{remark}
The next two propositions give some information about the structure of $\E$-local martingales and $(\E,\ZN)$-martingales. The results are almost literally taken from Corollaries 3.16 and 3.17 in \cite{CKS98}; the proofs there still work for our generalisation.
\begin{prop}
\label{propelocalmart}
Let $Y$ be a special semimartingale and $Y=Y_0+M^Y+ A^Y$ its canonical decomposition. Then $Y$ is an $\E$-local martingale if and only if $[M^Y,N]$ is locally \hbox{$P$-integrable} and $A^Y=-\la M^Y,N\ra$.
\end{prop}
\begin{prop}
\label{propemart}
A semimartingale of the form $Y=Y_0+M^Y-\la M^Y,N\ra$ and satisfying \hbox{$E\big[Y^*_T \big( \ZNTm {}^{T_m}\E(N) \big)^*_T\big]<\infty$} for any $m\in\N$ is an $(\E,\ZN)$-martingale.
\end{prop}
We also need the following definitions.
\begin{defi}
We say that $(\E,\ZN)$ with $\E=\E(N)$ is \emph{regular and square-integrable} if $\one_{\rrbracket T_m,T\rrbracket}\sint( \ZNTm {}^{T_m}\E(N))$ is a square-integrable (true) $P$-martingale and $\ZNTm$ is square-integrable 
for any $m$.
\end{defi}
\bl\label{lemart}
Suppose $(\E,\ZN)$ with $\E=\E(N)$ is regular and square-integrable. Let $(X^n)_{n\in\N}$ be a sequence of $(\E,\ZN)$-martingales with $X^n_T\in  \LiiP$ and $X^n_T\to H$ in $ \LiiP$ as $n\to\infty$. Then there exist a subsequence $(X^{n_\ell})_{\ell\in\N}$ and an $\E$-local martingale $X$ given by $X_T=H$ and 
\begin{equation}
\label{Xdef}
X_t:=\frac{E[H\, {}^{T_m}\E(N)_T|\F_t]}{{}^{T_m}\E(N)_t}
\quad\text{on $\llbracket T_m,T_{m+1} \llbracket$}
\end{equation}
such that $X^{n_\ell}\to X$ in the semimartingale topology (in $\cS(P)$, for short) as $\ell\to\infty$. If $\E(N)$ satisfies the reverse H\"older inequality $R_1(P)$, then $X$ is an $(\E,\ZN)$-mar\-tin\-gale.
\el
\bp
1) To show that $X$ above is an $\E$-local martingale with $X_T=H$, we argue similarly as in the proof of Proposition 3.12.iii) in \cite{CKS98}. More precisely, we exploit that we need not assume $\E(N)$ to satisfy $R_{q}(P)$ with $q=2$ as used there; it is sufficient to exploit that $\E(N)$ always satisfies $R_1(P)$ in a local sense. We define for each $m\in\N_0$ a sequence of stopping times $\tau^m_k=T_m\one_{F_k^c}+T\one_{F_k}$ with $F_k:=\big\{E\big[| {}^{T_m}\E(N)_T|\big|\F_{T_m}\big]\leq k\big\}$ for $k\in\N$. Then we rewrite \eqref{Xdef} after multiplication with $\ZNTm$ as
\begin{equation}
L_t := X_t \ZNTm {}^{T_m} \E(N)_t = E[ X_T \ZNTm {}^{T_m} \E(N)_T | \F_t]
\quad\text{on $\llbracket T_m, T_{m+1} \llbracket$}
\label{mprop}
\end{equation}
and note that the right-hand side is in $L^1(P)$ since $X_T=H$ and $\ZNTm {}^{T_m}\E(N)_T$ are both in $L^2(P)$. Hence
$
L_t \one_{\{T_m\le t<T_{m+1}\}}
$
is in $L^1(P)$ and so is then $X_{T_m} \ZNTm {}^{T_m} \E(N)_{T_m}$. To argue that $X$ is an $\E$-local martingale, we want to prove that $({}^{T_m}X \ZNTm {}^{T_m} \E(N))^{\tau^m_k}$ is a $P$-martingale, and \eqref{mprop} already gives the martingale property for the unstopped process ${}^{T_m} L$. So due to $\, {}^{T_m} X = X - X^{T_m}$, the $P$-integrability of $L_t$ and $\tau^m_k\ge T_m$, it only remains to show that
\begin{equation}
X_{T_m} \ZNTm {}^{T_m} \E(N)_{t\wedge\tau^m_k} \one_{\{T_m\le t<T_{m+1}\}} \in L^1(P) .
\label{prodint}
\end{equation}
But $\ZNTm {}^{T_m} \E(N)$ is a $P$-martingale and remains so after stopping by $\tau^m_k$, and the final value of that stopped process is
$$
\ZNTm {}^{T_m} \E(N)_{\tau^m_k}
=
\ZNTm {}^{T_m} \E(N)_{T_m} \one_{F_k^c} + \ZNTm {}^{T_m} \E(N)_T \one_{F_k}.
$$
Multiplying by $X_{T_m}$, conditioning on $\F_{T_m}$ and using the definition of $F_k$ hence gives \eqref{prodint}; indeed, we have
$$
E\big[ | X_{T_m} \ZNTm {}^{T_m} \E(N)_T \one_{F_k} | \big]
\le
E\big[ | X_{T_m} \ZNTm {}^{T_m} \E(N)_{T_m} | E[ |\E(N)_T| \,|\, \F_{T_m}] \one_{F_k} \big]
<\infty.
$$
This shows that $X$ is an $\E$-local martingale; and if $\E(N)$ satisfies $R_1(P)$, we have $F_k=\Om$, hence $\tau^m_k=T$, for $k$ large enough so that $X$ is even an $(\E,\ZN)$-martingale.

2) Now fix $m\in\N$ and take any subsequence of $(X^n)$, again denoted by $(X^n)$ in this step for ease of notation. Set $Y^{n,m} := {}^{T_m} X^n = X^n - (X^n)^{T_m}$ so that by the definition of
$(\E,\ZN)$-mar\-tin\-gales, the product of $\ZNTm {}^{T_m}\E(N)$ and $Y^{n,m}$ is a martingale. Note that $(Y^{n,m})^{\tau^m_k}=(X^n-(X^n)^{T_m})\one_{F_k}$ and $(Y^{m})^{\tau^m_k}=(X-X^{T_m})\one_{F_k}$ for each $k\in\N$. Since $X^n_T\to
X_T=H$ in $ \LiiP$ and
\be
X^n_t\, {}^{T_m}\E(N)_t=E[X^n_T\, {}^{T_m}\E(N)_T|\F_t]
\quad\text{on $\llbracket T_m,T_{m+1}\llbracket$}
\label{eq:p:emart}
\ee
for the $(\E,\ZN)$-martingales $X^n$ by \eqref{mprop}, we obtain for $n\to\infty$ that
\begin{align*}
E\big[|(X^n_{T_{m+1}\wedge \tau^m_k}-X_{T_{m+1}\wedge \tau^m_k})\ZNTm {}^{T_m}\E(N)_{T_{m+1}\wedge \tau^m_k}|\big]
&\leq E\big[|(X^n_T-H)\ZNTm {}^{T_m}\E(N)_{T}|\big]\\
& \leq \| X^n_T-H\|_{ \LiiP}\| \ZNTm {}^{T_m}\E(N)_{T}\|_{ \LiiP}
\end{align*}
tends to 0, and from the definition of $\tau_k^m$ that for $n\to\infty$,
\begin{align*}
&E\big[|(X^n_{T_{m+1}\wedge \tau^m_k}-X_{T_{m+1}\wedge \tau^m_k})\ZNTm {}^{T_m}\E(N)_{T_{m+1}\wedge \tau^m_k}|\big]\\
&=E\big[|E[(X^n_T-H)\, {}^{T_m}\E(N)_T|\F_{T_m}]\ZNTm {}^{T_m}\E(N)_{T_{m+1}\wedge \tau^m_k}|\big]\\
&\leq E\big[ |(X^n_T-H)\ZNTm {}^{T_m}\E(N)_{T}| \big] k
\longrightarrow 0.
\end{align*}
This gives $\ZNTm {}^{T_m}\E(N)_{T\wedge \tau^m_k} \, Y^{n,m}_{T\wedge \tau^m_k} \to \ZNTm{}^{T_m}\E(N)_{T\wedge \tau^m_k} \, Y^{m}_{T\wedge \tau^m_k}$ in $L^1(P)$ as $n\to\infty$ because ${}^{T_m}\E(N)_T=0$ on $\{T_{m+1}<T\}$. Theorem 4.21 in \cite{J79} then yields a subsequence
$(Y^{n_j,m})_{j\in\N}$ such that
$$\big( \ZNTm {}^{T_m}\E(N) \, Y^{n_j,m}\big)^{\tau^m_k}\longrightarrow \big( \ZNTm{}^{T_m}\E(N) \, Y^{m}\big)^{\tau^m_k}
\quad\text{locally in $\underline{\underline{H}}_{\,\rm loc}^1(P)$ as $j\to\infty$}$$
and therefore $\ZNTm {}^{T_m}\E(N)\, Y^{n_j,m} \to \ZNTm{}^{T_m}\E(N) \, Y^m$ in $\cS(P)$ as $j\to\infty$ by Theorem V.14 in \cite{P04}. Because $\frac{1}{\ZNTm {}^{T_m}\E(N)}\one_{\llbracket T_m,T_{m+1} \llbracket}$ is a semimartingale and the multiplication of semimartingales is continuous in $\cS(P)$, we get $Y^{n_j,m}\one_{\llbracket T_m,T_{m+1} \llbracket} \to Y^m\one_{\llbracket T_m,T_{m+1} \llbracket}$ in $\cS(P)$ as $j\to\infty$. Note that the subsequence $(n_j)_{j\in\N}$ depends on $m$. 

3) Now we construct the desired subsequence $(n_\ell)_{\ell\in\N}$ by a diagonal argument, as follows. Start with $m=0$ and the original sequence $(X^n)$ to obtain from step 2) a subsequence $(n_j(0))_{j\in\N}$, and take $n_1 := n_1(0)$. Then take $m=1$, apply step 2) for the subsequence $(X^{n_j(0)})_{j\in\N}$ to get a new subsequence $(n_j(1))_{j\in\N}$, and take $n_2 := n_1(1)$. Iterating this procedure yields our subsequence $(n_\ell)_{\ell\in\N}$, and we claim that $X^{n_\ell}\to X$ in $\cS(P)$ as $\ell\to\infty$. To see this, use the definition of $Y^{n,m}$ to write
\be
X^{n_\ell}=\sum_{m=0}^\infty Y^{n_\ell}\one_{\llbracket T_m,T_{m+1}\llbracket}+\sum_{m=0}^\infty
X^{n_\ell}_{T_m}\one_{\llbracket T_m,T_{m+1}\llbracket}+ X^{n_\ell}_T\one_{\llbracket T\rrbracket}
.\label{eq2:p:emart}
\ee
Since $Y^{{n_j(m)},m}\one_{\llbracket T_m,T_{m+1}\llbracket} \to Y^m\one_{\llbracket T_m,T_{m+1}\llbracket}$ as
$j\to\infty$, the first sum converges in $\cS(P)$ to
$$
\sum_{m=0}^\infty Y^{m}\one_{\llbracket T_m,T_{m+1}\llbracket}
=
X-\sum_{m=0}^\infty X_{T_m}\one_{\llbracket T_m,T_{m+1} \llbracket}-X_T\one_{\llbracket T\rrbracket} ,
$$
where the equality now uses the definition of $Y^m = X - X^{T_m}$. To obtain the convergence of the second sum in \eqref{eq2:p:emart}, we
observe that
\begin{align*}
E\big[ |X^n_{T_m}-X_{T_m}| \ZNTm \big]
&=
E\big[|E[(X^n_T-H)\, {}^{T_m}\E(N)_{T}|\F_{T_m}]| \ZNTm \big]
\\
&\leq \|X^n_T-H\|_{ \LiiP}\| \ZNTm {}^{T_m}\E(N)_{T}\|_{ \LiiP}
\end{align*}
by \eqref{eq:p:emart} for all $m\in\N_0$ and for $m=\infty$ with $T_\infty:=T$ and therefore
as $\ell\to\infty$,
\begin{equation}
\label{aux:conv}
\sum_{m=0}^\infty \ZNTm X^{n_\ell}_{T_m}\one_{\llbracket T_m,T_{m+1}\llbracket}+ X^{n_\ell}_T\one_{\llbracket T\rrbracket}
\longrightarrow
\sum_{m=0}^\infty \ZNTm X_{T_m}\one_{\llbracket T_m,T_{m+1} \llbracket}+X_T\one_{\llbracket T\rrbracket}
\end{equation}
locally in $\underline{\underline{H}}^1(P)$ with the localising sequence $(T_m)$. As local convergence in
$\underline{\underline{H}}^1(P)$ implies convergence in $\cS(P)$ again by Theorem V.14 in \cite{P04}, \eqref{aux:conv} also holds in $\cS(P)$. Because 
$
\sum_{m=0}^\infty \frac{1}{\ZNTm} \one_{\llbracket T_m, T_{m+1} \llbracket}
$
is a semimartingale and the multiplication of semimartingales is continuous in $\cS(P)$, this completes the proof.
\ep
%
%\marginpar{coremart}
\begin{cor}\label{coremart}
Suppose that $(\E,\ZN)$ with $\E=\E(N)$ is regular and square-integrable, $S=S_0+M-\la M, N\ra$ is in $\HzlocP$ and
$(\vt^n)_{n\in\N}$ is a sequence in $\TS$ such that $\vt^n\sint S_T\to H$ in $ \LiiP$. Then $\vt^n\sint S$ is an
$(\E,\ZN)$-martingale for each $n\in\N$, and there exist $\vt\in\TB$ with $\vt\sint S_T=H$ and
$$
\vt\sint S_t=\frac{E[ (\vt\sint S_T) \, {}^{T_m}\E(N)_T|\F_t]}{{}^{T_m}\E(N)_t}=\frac{E[H\, {}^{T_m}\E(N)_T|\F_t]}{{}^{T_m}\E(N)_t}
\quad\text{on $\llbracket T_m,T_{m+1} \llbracket$}
$$
and a subsequence $(\vt^{n_k})_{k\in\N}$ in $\TS$ such that $\vt^{n_k}\sint S\to\vt\sint S$ in $\cS(P)$ as
$k\to\infty$.
\end{cor}
\bp
By Proposition \ref{propemart}, $S$ is an $\E$-local martingale and $\vt^n\sint S$ is an $(\E,\ZN)$-martingale for each $n$. Then Lemma \ref{lemart} gives the existence of an $\E$-local martingale $X$ and a subsequence $(\vt^{n_k})$ in $\TS$ such that $X_T=H$ and $X_t=\frac{E[H\, {}^{T_m}\E(N)_T|\F_t]}{{}^{T_m}\E(N)_t}$ on $\llbracket T_m,T_{m+1} \llbracket$ and $\vt^{n_k}\sint S\to X$ in $\cS(P)$. As the space of stochastic integrals is closed under convergence in $\cS(P)$ by Theorem V.4 in \cite{Memin}, there exists some $\vt\in \cL(S)$ with $\vt\sint S=X$. Since convergence in $\cS(P)$ implies ucp-convergence and therefore that $\vt^{n_k}\sint S_\tau \to \vt\sint S_\tau$ in probability for all
stopping times $\tau$, we obtain that $\vt\in\TB$ which completes the proof.
\ep
To deal with the fact that different integrands may lead to the same stochastic integral (or, in financial terms, that we may have redundant assets), we introduce the \emph{projection on the predictable range}. For a detailed explanation of the related issues of selecting particular representatives of equivalence classes of integrands as well as for sufficient conditions for the closedness of the projection on the predictable range for certain correspondences, we refer the reader to \cite{CS11}.
\begin{prop}
For each $\RR^d$-valued semimartingale $Y$, there exists an $\R^{d \times d}$-valued predictable process $\Pi^Y$, called the \emph{projection on the predictable range of $Y$}, which takes values in the orthogonal projections in $\R^d$ and has the following property: If $\vt$ is in $\cL(Y)$ and $\vp$ is predictable, then $\vp$ is in $\cL(Y)$ with $\vp\sint Y = \vt\sint Y$ (up to indistinguishability) if and only if $\Pi^Y\vt=\Pi^Y\vp$ $P_B$-a.e. We choose and fix one version of $\Pi^Y$.
\end{prop}
\bp
See Lemma 5.3 in \cite{CS11}.
\ep
\begin{example}
\label{Exa:Ito}
For the frequently used It\^o process models of the form
$$
\frac{dY^i_t}{Y^i_t} = (\mu^i_t - r_t) \, dt + \sum_{k=1}^m \sigma^{ik}_t \, dW^k_t,
$$
$\Pi^Y$ is the projection on the orthogonal complement of the kernel of $\sigma\sigma^\top$. If each $\sigma_t \sigma^\top_t$ is invertible (as is usually assumed), $\Pi^Y$ is just the identity. This holds in particular when $m=d$ and each $\sigma_t$ is invertible, i.e.~when the model is complete without the constraints.
\end{example}

After these preparations, we obtain the closedness of $G_T(\TBC)$ by the following theorem. We recall that this implies the \emph{existence of a solution} to the constrained mean-variance hedging problem \eqref{MVH}, for any payoff $H\in  \LiiP$, if $C$ has also convex values.

\bt\label{thmclosedTBC}
Suppose that $(\E,\ZN)$ with $\E=\E(N)$ is regular and square-integrable and $S=S_0+M-\la M, N\ra$ is in $\HzlocP$ so that $S$ is an $\E$-local martingale by Proposition \ref{propelocalmart}. Let $C:\Ombar\to 2^{\R^d}\setminus\{\emptyset\}$ be a predictable correspondence with closed values such that the projection of $C$ on the predictable range of $S$ is closed, i.e. $\Pi^S\omt C\omt$ is $P_B$-a.e.~closed. Then $G_T(\TBC)$ is closed in $\LiiP$.
\et
\bp
Let $(\vt^n)$ be a sequence in $\TBC$ with $\vt^n\sint S_T\to H$ in $\LiiP$. Using the definition of $\TBC$ and a diagonal argument yields a sequence $(\vp^n)$ in $\Theta(C)$ with $\vp^n\sint S_T\to H$ in $\LiiP$. Then Corollary \ref{coremart} implies that there exist $\vt\in\TB$ with $\vt\sint S_T=H$ and a subsequence, again indexed by $n$, with $\vp^n\sint S\to \vt\sint S$ in $\cS(P)$. Since $\cC\sint S=\{\psi\sint S \,|\, \psi\in\cC\}$ is closed in $\cS(P)$ by Theorem 4.5 in \cite{CS11}, the integrand $\vt$ can be chosen
$C$-valued; this uses the assumption on $\Pi^S C$. As convergence in $\cS(P)$ implies ucp-convergence, we obtain
$\vp^n\sint S_\tau\to
\vt\sint S_\tau$ in probability for all stopping times $\tau$, and therefore $\vt$ is in $\TBC$. This completes the
proof.
\ep
\begin{remark}\label{compare}
Let us briefly compare Theorem \ref{thmclosedTBC} to the main result of Theorem 3.12 in \cite{CS10a}. The latter imposes the extra assumption that $\E(N)$ satisfies the reverse H\"older inequality $R_2(P)$, and proves that the space $G_T(\TSC)$ is then closed in $\LiiP$. So Theorem \ref{thmclosedTBC} here has a weaker assumption; but since $\TBC$ is bigger than the space $\TSC$ considered in \cite{CS10a}, one also feels it could be easier for $G_T(\TBC)$ to be closed in $ \LiiP$.
\end{remark}
Having clarified the existence of a solution to \eqref{MVH} or \eqref{ap-1}, our goal in the sequel is to describe its \emph{structure} in more detail. This is done via stochastic control techniques and in particular dynamic programming, and for that, we need certain properties for the space $\TBC$ of strategies we work with. This is the reason why we slightly changed the definition in comparison to \cite{CK07}: We want to show, without assuming that there exists an ELMM $Q$ for $S$ with $\frac{dQ}{dP}\in\LiiP$, that $\TBC$ is stable under bifurcation and almost stable. 
%
%\marginpar{las}
\bl\label{las}
For any predictable correspondence $C:\Ombar\to 2^{\R^d}\setminus\{\emptyset\}$, the space $\TBC$ has the following properties:

\vspace{0.25\baselineskip}
% \bi
% \item[
{\bf1) }$\TBC$ is \emph{stable under bifurcation}: If $\vt, \varphi$ are in $\TBC$, $\sigma$ is a stopping time, $F\in\F_\sigma$ and $\vartheta \one_{\llbracket 0,\sigma\rrbracket }=\varphi \one_{\llbracket 0,\sigma\rrbracket }$, then $\psi=\vt\one_F+\varphi\one_{F^c}$ is also in $\TBC$.

\vspace{0.25\baselineskip}
% \item[
{\bf 2) }$\TBC$ is \emph{almost stable}: For all $\vt, \varphi$ in $\TBC$, stopping times $\sigma$ and $F\in\F_{\sigma}$ with $P[F]>0$, there is for each $\varepsilon\in(0,P[F])$ a
set
$F_{\varepsilon}\subseteq F$ in $\F_{\sigma}$ with $P[F\setminus F_{\varepsilon}]\leq\varepsilon$ such that
$$
\hbox{$\psi := \vt \one_{F_{\varepsilon}^{c}}+(\vt \one_{\llbracket 0,\sigma\rrbracket }+\varphi \one_{\rrbracket \sigma,T\rrbracket })\one_{F_{\varepsilon}}$ is in $\TBC$}
$$
and $\vt\sint S_{\sigma}$ is uniformly bounded on $F_{\varepsilon}$.
% \ei
\el
\begin{proof}
By the definition of $\TBC$, we must in both cases find a sequence $(\psi^n)$ in $\Theta(C)$ such that $\psi^n\sint S_T \overset{ \LiiP}{\longrightarrow} \psi\sint S_T$ and $\psi^n\sint S_\tau\overset{P}{\longrightarrow}\psi\sint S_\tau$ for all stopping times $\tau$. We start with approximating sequences $(\vt^n)$ and $(\varphi^n)$ in $\Theta(C)$ for $\vt,\vp\in\TBC$.

1) For $\psi^n:=\vt^n\one_F+\varphi^n\one_{F^c}\in\Theta(C)$, the local character of stochastic integrals yields
$$\|\psi^n\sint S_T-\psi\sint S_T\|_{ \LiiP}=\|(\vt^n\sint S_T-\vt\sint S_T)\one_F+(\varphi^n\sint S_T-\varphi\sint S_T)\one_{F^c}\|_{ \LiiP}\overset{n\to\infty}{\longrightarrow}0$$
and, for all stopping times $\tau$,
$$\psi^n\sint S_\tau=(\vt^n\sint S_\tau)\one_F+(\varphi^n\sint S_\tau)\one_{F^c}\overset{P}{\longrightarrow}(\vt\sint S_\tau)\one_F+(\varphi\sint S_\tau)\one_{F^c}=\psi\sint S_\tau.
$$

2) By Egorov's theorem, we can find for each $\varepsilon\in(0,P[F])$ a set $F_{\varepsilon}\in\F_{\sigma}$ with $P[F\setminus F_{\varepsilon}]\leq \varepsilon$ such that $\vt^n\sint S_\sigma\to\vt\sint S_\sigma$ and $\varphi^n\sint S_\sigma\to\varphi\sint S_\sigma$ uniformly on $F_{\varepsilon}$. For the sequence $\psi^n := \vt^n\one_{F_{\varepsilon}^c}+(\vt^n\one_{\llbracket 0,\sigma\rrbracket }+\varphi^n\one_{\rrbracket \sigma,T\rrbracket })\one_{ F_{\varepsilon}}$ in $\Theta(C)$, we obtain again from the local character of stochastic integrals that
\begin{align*}
&\big\|\psi^n\sint S_T-\big(\vt \one_{F_{\varepsilon}^{c}}+(\vt \one_{\llbracket 0,\sigma\rrbracket }+\varphi \one_{\rrbracket \sigma,T\rrbracket })\one_{F_{\varepsilon}}\big)\sint S_T\big\|_{ \LiiP}\\
&\leq \|(\vt^n\sint S_T-\vt\sint S_T)\one_{ F_{\varepsilon}^{c}}\|_{ \LiiP}+\|(\vt^n\sint S_\sigma-\vt\sint
S_\sigma)\one_{ F_{\varepsilon}}\|_{ \LiiP}\\ 
&\phantom{\le\ }+\|(\varphi^n\sint S_{\sigma}-\varphi\sint S_{\sigma})\one_{
F_{\varepsilon}}\|_{ \LiiP}+\|(\varphi^n\sint S_T-\varphi\sint S_T)\one_{
F_{\varepsilon}}\|_{ \LiiP}\overset{n\to\infty}\longrightarrow0,
\end{align*}
where the first and the last term on the right-hand side converge to zero by the choice of $(\vt^n)$ and $(\varphi^n)$ and the two middle terms by the uniform convergence on $ F_{\varepsilon}$. Since
\begin{align*}
\psi^n\sint S_\tau
&=
(\vt^n\sint S_\tau)\one_{ F_{\varepsilon}^{c}}+(\vt^n\sint S_{\sigma\wedge\tau})\one_{ F_{\varepsilon}}+(\varphi^n\sint S_\tau-\varphi^n\sint S_{\sigma\wedge\tau})\one_{F_{\varepsilon}} , 
\\
\psi\sint S_\tau
&=
(\vt\sint S_\tau)\one_{F_{\varepsilon}^{c}}+(\vt\sint S_{\sigma\wedge\tau})\one_{F_{\varepsilon}}+(\varphi\sint S_\tau-\varphi\sint S_{\sigma\wedge\tau})\one_{F_{\varepsilon}}
\end{align*}
for all stopping times $\tau$ again by the local character of stochastic integrals, we obtain that $\psi^n\sint S_\tau\overset{P}{\longrightarrow}\psi\sint S_\tau$ for all stopping times $\tau$.

Finally, to get $\vt\sint S_\sigma$ uniformly bounded on $F_{\ve}$ as well, one starts instead of $F$ with some $F_N':=F\cap\{|\vt\sint S_\sigma|\leq N\}\in\F_\sigma$. Then $F_N'\nearrow F$, so $P[F_N']$ increases to $P[F]$ as $N\to\infty$, and taking $N(\ve)$ large enough will give the result. This completes the proof.
\end{proof}
%%%%%%%%%%%%%%%%%%%%%%%%%%%%%%%%%%%%%%%%%%%%%%%%%%%%%%%%%%%%%%%%%%%%%%%%%%%%%%%%%%%%%%%%
%\newpage
\section{Dynamic programming}\label{sec:dp}
In this section, we establish a dynamic description of the optimal strategy for \eqref{ap-1} by dynamic programming. To that end, we consider the problem to
\be
\text{minimise $E\big[ |V_T(x,\vt)|^2 \big] = E\big[ |x+\vt\sint S_T|^2 \big]$ over all $\vt\in\TBK$}\label{ap}
\ee
for a fixed $x\in\R$ and a predictable correspondence $K:\Ombar\to 2^{\RR^d}\setminus\{\emptyset\}$ with \emph{closed cones} as values. We view \eqref{ap} as a stochastic optimal control problem and want to study the corresponding value process.

We first need some notation. For any stopping time $\tau$ with values in $[0,T]$, we denote by $\cS_{\tau,T}$ the family of all stopping times $\sigma$ with $\tau\leq\sigma\leq T$ (so that $\tau\in\cS_{0,T}$). In order to describe the optimisation starting at time $\tau$ with wealth $x$, we define for $\tau\in\cS_{0,T}$, $\sigma\in\cS_{\tau,T}$ and $\vt\in\TBK$ with $\vt=0$ on $\llbracket 0, \tau \rrbracket$ the space
\begin{align*}
\fK(\vt,\sigma;\tau)
&:=
\big\{\vp\in\TBK \,\big|\, \vp=0 \text{ on $\llbracket0,\tau\rrbracket$ and $\vp\one_{\rrbracket\tau,\sigma\rrbracket}=\vt\one_{\rrbracket\tau,\sigma\rrbracket}$}\big\}
\\
&\phantom{:}=
\big\{\vp\in\TBK \,\big|\, \vp\one_{\llbracket 0,\sigma\rrbracket}=\vt\one_{\llbracket 0,\sigma\rrbracket}\big\}.
\end{align*}
Note that $\fK(\vt,\sigma;\sigma) = \fK(0,\sigma;\sigma)$. We then define for $\vp\in\fK(\vt,\sigma;\tau)$ the random variables
$$
\textstyle
\Gamma(\vp,\sigma;x,\tau,\vt)
:=
E\big[|V_T(x,\vp)|^2\big|\F_\sigma\big]
=
E\big[ |x+\int_\tau^\sigma \vt_u \,dS_u+\int_\sigma^T \vp_u \,dS_u|^2\big|\F_\sigma\big] ,
$$
and for $\sigma\in\cS_{\tau,T}$ and $\vt\in\TBK$ with $\vt=0$ on $\llbracket0,\tau\rrbracket$
$$
\bJ(\sigma;x,\tau,\vt):=\underset{\vp\in\fK(\vartheta,\sigma;\tau)}{\essinf}\Gamma(\vp,\sigma;x,\tau,\vt).
$$
Because the family $\{\Gamma(\vp,\sigma;x,\tau,\vt) \,|\, \vp\in\fK(\vartheta,\sigma;\tau)\}$ is stable under taking minima by part 1) of Lemma \ref{las}, the family $\{\bJ(\sigma;x,\tau,\vt) \,|\, \sigma\in\cS_{\tau,T}\}$ for any fixed $\tau\in\cS_{0,T}$ is a submartingale system for any $\vt\in\TBK$ with $\vt=0$ on $\llbracket 0,\tau\rrbracket$. It is a martingale system for $\tvt\in\TBK$ with $\tvt=0$ on $\llbracket 0,\tau\rrbracket$ if and only if $\tvt=\tvp^{(x,\tau)}$ is optimal for the problem to
\be
\textstyle
\text{minimise $E\big[ |x+\int_\tau^T\vp_u \,dS_u|^2 \big] = E \big[ |x+\vp\sint S_T|^2 \big]$ over all $\vp\in\fK(0,\tau;\tau)$.}\label{A}
\ee
These facts follow by standard arguments as e.g.~in Chapter 1 of \cite{EK79} or the proof of Theorem 4.1 in \cite{LP99}. We now exploit the quadratic and conic structure of our problem to obtain a decomposition of $\bJ$.
\begin{prop}\label{pJL}
For any stopping time $\tau\in\cS_{0,T}$, there exist families of random variables $\{\bL^\pm(\sigma) \,|\, \sigma\in\cS_{\tau,T}\}$ such that
\begin{align}
\bar J(\sigma;x,\tau,\vartheta)
&=
\textstyle
\underset{\vp\in\fK(\vartheta,\sigma;\tau)}{\essinf}E\big[ |x+\int_\tau^\sigma \vt_u \,dS_u+\int_\sigma^T \vp_u \,dS_u|^2 \big|\F_\sigma\big] \nonumber\\
&=
\textstyle
\big((x+\int_\tau^\sigma \vt_u \,dS_u)^+\big)^2\bL^+(\sigma)+\big((x+\int_\tau^\sigma \vt_u \,dS_u)^-\big)^2\bL^-(\sigma)\label{B}
\end{align}
for any $\sigma\in\cS_{\tau,T}$ and any $\vt\in\TBK$ with $\vt=0$ on $\llbracket 0,\tau\rrbracket$. The random variables $\bL^\pm(\sigma)$ do not depend on $x$, $\tau$ or $\vt$ and are explicitly given by
\be
\textstyle
\bL^\pm(\sigma):=\underset{\vp\in\fK(0,\sigma;\sigma)}{\essinf}E\big[|1\pm\int_\sigma^T \vp_u \,dS_u|^2\big|\F_\sigma\big]=\bJ(\sigma;\pm 1,\sigma,0).\label{C}
\ee
In particular, all the $\bL^\pm(\sigma)$ are $[0,1]$-valued, and $\bL^\pm(T)=1$.
\end{prop}
\bp
Fix $x,\tau,\vt$ and $\sigma$ and define $\bL^\pm(\sigma)$ by \eqref{C}. The last assertion is then obvious, and the intuition for \eqref{B} is that the quadratic structure of our problem and the fact that the constraints are given by cones allow us to pull out an $\F_\sigma$-measurable factor. Note that we can also write $\vt\sint S_\sigma$ instead of $\int_\tau^\sigma \vt_u \,dS_u$ because $\vt=0$ on $\llbracket 0,\tau\rrbracket$. For the detailed proof of \eqref{B}, we argue by contradiction. Suppose first that
$$
\bJ(\sigma;x,\tau,\vartheta)<\big((x+\vt\sint S_\sigma)^+\big)^2\bL^+(\sigma)+\big((x+\vt\sint S_\sigma)^-\big)^2\bL^-(\sigma)
\quad\text{on $F^\prime$}
$$
for some set $F^\prime\in\F_\sigma$ with $P[F^\prime]>0$. Then there exist $\vp\in\fK(\vartheta,\sigma;\tau)$ and $F\in\F_\sigma$ with $F\subseteq F^\prime$  and $P[F]>0$ such that
\be
E\big[ |x+\vp\sint S_T|^2\big|\F_\sigma\big]<\big((x+\vt\sint S_\sigma)^+\big)^2\bL^+(\sigma)+\big((x+\vt\sint S_\sigma)^-\big)^2\bL^-(\sigma)
\quad\text{on $F$.}
\label{eq:ineq1}
\ee
Since $\bar J(\sigma;x,\tau,\vartheta)\geq 0$, we have $F\subseteq \{0< |x+\vt\sint S_\sigma|\}$ and can write
\begin{align}
E\big[ |x+\vp\sint S_T|^2\big|\F_\sigma\big]
&=
\big((x+\vt\sint S_\sigma)^+\big)^2E\bigg[\left(1+\frac{\one_{\rrbracket\sigma,T\rrbracket }\vp}{(x+\vt\sint S_\sigma)^+}\sint S_T\right)^2\bigg|\F_\sigma\bigg]
\nonumber
\\
&\phantom{=\ }
+\big((x+\vt\sint S_\sigma)^-\big)^2E\bigg[\left(1-\frac{\one_{\rrbracket \sigma,T\rrbracket }\vp}{(x+\vt\sint S_\sigma)^-}\sint S_T\right)^2\bigg|\F_\sigma\bigg]
\quad\text{on $F$.}
\label{eq-sep}
%<\big((1-\vt\sint S_\sigma)^+\big)^2L_\sigma^++\big((1-\vt\sint S_\sigma)^-\big)^2L_\sigma^-
\end{align}
Plugging the last expression into \eqref{eq:ineq1}, we obtain
$$
E\bigg[\left(1\pm \frac{\one_{\rrbracket \sigma,T\rrbracket }\vp}{(x+\vt\sint S_\sigma)^\pm}\sint S_T\right)^2\bigg|\F_\sigma\bigg]
< \bL^\pm(\sigma)
\quad\text{on $F^\pm:=F\cap\{x+\vt\sint S_\sigma\gtrless 0\}$.}
$$
To derive a contradiction to the definition of $\bL^\pm(\sigma)$, it remains to show that
$$
\psi^\pm:=\frac{\one_{\rrbracket \sigma,T\rrbracket }\vp}{(x+\vt\sint S_\sigma)^\pm}\one_{G^\pm}\in\fK(0,\sigma;\sigma)
$$
for some sets $G^\pm\in\F_\sigma$ with $G^\pm\subseteq F^\pm$ and $P[G^\pm]>0$. To that end, let $(\vp^n)_{n\in\N}$ be an approximating sequence in $\TS(K)$ for $\vp$. By passing to a subsequence again indexed by $n$, we can assume that $\vp^n\sint S_\sigma\to\vp\sint S_\sigma$ $P$-a.s. Then we can find $G^+\in\F_\sigma$ with $G^+\subseteq F^+$ and $P[G^+]>0$ such that $m\geq|x+\vt\sint S_\sigma|\geq \frac{1}{m}$ on $G^+$ for some $m\in\N$, by continuity of $P$ from below, and $\vp^n\sint S_\sigma\to\vp\sint S_\sigma$ uniformly on $G^+$, by Egorov's theorem. Moreover, we obtain that $\psi^n:=\frac{\one_{\rrbracket \sigma,T\rrbracket }\vp^n}{(x+\vt\ssint S_\sigma)^+}\one_{G^+}\in\TS(K)$ because $K$ is cone-valued, and
$$
|\psi^n\sint S_\vr-\psi^+\sint S_\vr|
\leq 
(|\vp^n\sint S_\vr-\vp\sint S_\vr|+|\vp^n\sint S_\sigma-\vp\sint S_\sigma|)\frac{1}{m}\one_{G^+}
$$
for all stopping times $\vr$. By the choice of $(\vp^n)$ and the local character of stochastic integrals, the right-hand side converges to zero in probability for all stopping times $\vr$, and in $ \LiiP$ for $\vr=T$. Since $\psi^n\sint S=0=\psi^+\sint S$ on $\llbracket0,\tau\rrbracket$, we have that $\psi^+\in\fK(0,\sigma;\sigma)$. By analogous arguments, we can also establish that $\psi^-\in\fK(0,\sigma;\sigma)$.

To complete the proof of \eqref{B}, we now assume that
$$
\bar J(\sigma;x,\tau,\vartheta)>\big((x+\vt\sint S_\sigma)^+\big)^2\bL^+(\sigma)+\big((x+\vt\sint S_\sigma)^-\big)^2 \bL^-(\sigma)
\quad\text{on $F$}
$$
for some set $F\in\F_\sigma$ with $P[F]>0$. Then there exist $\vp^+$ and $\vp^-$ in $\fK(0,\sigma;\sigma)$, some $\ve>0$ and $F_\varepsilon\in\F_\sigma$ with $F_\varepsilon\subseteq F$ and $P[F_\varepsilon]>2\varepsilon$ such that
\begin{align}
\bar J(\sigma;x,\tau,\vartheta)
&\geq{}
\big((x+\vt\sint S_\sigma)^+\big)^2E\big[ |1+\vp^+\sint S_T|^2\big|\F_\sigma\big] \nonumber\\
&\phantom{\ge\ }
+\big((x+\vt\sint S_\sigma)^-\big)^2E\big[ |1-\vp^-\sint S_T|^2\big|\F_\sigma\big]+2\ve
\quad\text{on $F_\ve$.}\label{eq:pr:eqJL1}
\end{align}
By the definition of the essential infimum, there exists $\vp^\varepsilon\in\fK(\vartheta,\sigma;\tau)$ such that
\be
E\big[ |x+\vp^\ve\sint S_T|^2\big] < E\big[ \bar J(\sigma;x,\tau,\vartheta) \big]+\varepsilon^2.\label{eq:pr:eqJL2}
\ee
Since $\{|x+\vt\sint S_\sigma|\leq m\}\nearrow \Om$ for $m\to\infty$, there exists $G_\varepsilon\in\F_\sigma$ with $G_\varepsilon\subseteq F_\varepsilon$ and $P[G_\varepsilon]>\varepsilon$ and such that $|x+\vt\sint S_\sigma|\leq m$ on $G_\varepsilon$, and therefore
$$
\chi:=\big((x+\vt\sint S_\sigma)^+\vp^++(x+\vt\sint S_\sigma)^-\vp^-\big)\one_{G_\ve}\in\fK(0,\sigma;\sigma)
$$
by the local character of stochastic integrals. Moreover, we can by part 2) of Lemma \ref{las} without loss of generality choose $G_\ve$ such that $\psi := \vp^\ve \one_{G_{\varepsilon}^{c}}+(\vt \one_{\llbracket 0,\sigma\rrbracket }+\chi \one_{\rrbracket \sigma,T\rrbracket })\one_{G_{\varepsilon}}$ is in $\fK(\vartheta,\sigma;\tau)$. Then we use that $\vp^\ve\in\fK(\vartheta,\sigma;\tau)$, the definitions of $\psi$ and $\chi$, and \eqref{eq:pr:eqJL1} to write 
\begin{align*}
E\big[ |x+\vp^\ve\sint S_T|^2\big|\F_\sigma\big]
&\geq \one_{G_{\varepsilon}^c}E\big[ |x+\psi\sint S_T|^2\big|\F_\sigma\big]+\one_{G_{\varepsilon}}\bJ(\sigma; x, \tau,\vt)\\
&\geq \one_{G_{\varepsilon}^c}E\big[ |x+\psi\sint S_T|^2\big|\F_\sigma\big]+\one_{G_{\varepsilon}}\big(E\big[ |x+\vt\sint S_\sigma+\chi\sint S_T|^2 \big| \F_\sigma \big]+2\varepsilon\big) .
\end{align*}
In view of \eqref{eq:pr:eqJL2}, the definition of $\psi$ and since $P[G_\varepsilon]>\varepsilon$ and $\psi\in\fK(\vartheta,\sigma;\tau)$, we obtain after taking expectations that
$$
E\big[ \bar J(\sigma;x,\tau,\vartheta) \big]
>E\big[ |x+\vp^\ve\sint S_T|^2\big]-\varepsilon^2
\geq E\big[ |x+\psi\sint S_T|^2\big]+2\varepsilon^2-\varepsilon^2
\geq E\big[ \bar J(\sigma;x,\tau,\vartheta) \big]+\ve^2
$$
which is a contradiction. So \eqref{B} must hold.
\ep
Our next result shows that the random variables $\bL^\pm(\sigma)$ as well as $\bJ(\sigma;x,\tau,\vt)$ can be aggregated into nice RCLL processes.
\begin{prop}\label{prop:mop}
% \bi
% \item[
{\bf 1) }There exist RCLL submartingales $(L^\pm_t)_{0\leq t\leq T}$, called \emph{opportunity processes}, such that
\be
L^\pm_\sigma=\bL^\pm(\sigma)
\quad\text{$P$-a.s.~for each $\sigma\in\cS_{0,T}$.}\label{D}
\ee

\vspace{0.25\baselineskip}

{\bf 2) }Fix $x\in\R$ and $\tau\in\cS_{0,T}$. Define the RCLL process $(J_t(\vt;x,\tau))_{0\leq t\leq T}$ for every $\vt\in\TBK$ with $\vt=0$ on $\llbracket0,\tau\rrbracket$ by
\be
\textstyle
J_t(\vt;x,\tau)=\big((x+\int_\tau^t\vt_u\,dS_u)^+\big)^2L^+_t+\big((x+\int_\tau^t\vt_u\,dS_u)^-\big)^2L^-_t.\label{E}
\ee
Then we have for each $\vt\in\TBK$ with $\vt=0$ on $\llbracket0,\tau\rrbracket$ that
\be
J_\sigma(\vt;x,\tau)=\bJ(\sigma;x,\tau,\vt)
\quad\text{$P$-a.s.~for each $\sigma\in\cS_{\tau,T}$.}\label{F}
\ee
Moreover, $J(\vt;x,\tau)$ is a submartingale for every $\vt\in\TBK$ with $\vt=0$ on $\llbracket 0, \tau\rrbracket$, and $J(\tvt;x,\tau)$ is a martingale for $\tvt\in\TBK$ with $\tvt=0$ on $\llbracket 0,\tau\rrbracket$ if and only if $\tvt=\tvp^{(x,\tau)}$ is optimal for \eqref{A}.
\end{prop}
\bp
1) For $\tau\equiv0$, $(\bL^\pm(t))_{0\leq t\leq T}$ are submartingales by Proposition \ref{pJL}. They have by Theorem VI.4 in \cite{DM82} RCLL versions if the mappings \hbox{$t\mapsto E[\bL^\pm(t)]$} are right-continuous. We only prove this for $\bL^-$ as the argument for $\bL^+$ is completely analogous, but argue a bit more generally than directly needed. Fix a stopping time $\sigma\in\cS_{\tau,T}$. By \eqref{C} and the definition of the essential infimum, there exists for each $\varepsilon>0$ some $\vt^{\varepsilon}\in\fK(0,\sigma;\sigma)$ with
$$
E\big[ \bL^-(\sigma) \big]>E\big[ |1-\vt^\varepsilon\sint S_T|^2 \big]-\varepsilon,
$$
and $\vt^\ve$ can be chosen in $\Theta$ as the $\LiiP$-closure of $G_T(\Theta(K))$ contains $G_T(\TBK)$. Let $(\sigma_n)$ be a sequence in $\cS_{\sigma,T}$ with $\sigma_n\searrow \sigma$. Then $(\one_{\rrbracket \sigma_n,T\rrbracket }\vt^{\varepsilon})\sint S\overset{\HzP}\longrightarrow(\one_{\rrbracket \sigma,T\rrbracket }\vt^{\varepsilon})\sint S$ and thus $E[|1-(\one_{\rrbracket \sigma_n,T\rrbracket }\vt^\varepsilon)\sint S_T|^2]\to E[|1-(\one_{\rrbracket \sigma,T\rrbracket }\vt^\varepsilon)\sint S_T|^2]$ by Theorem IV.5 in \cite{P04}. Therefore
\begin{equation*}
E\big[ \bL^-(\sigma) \big]
>
\Lim_{n\to\infty}E\big[ |1-(\one_{\rrbracket \sigma_n,T\rrbracket }\vt^\varepsilon)\sint S_T|^2 \big]-\varepsilon
\geq
\Lim_{n\to\infty}E\big[ \bL^-(\sigma_n) \big]-\varepsilon,
\end{equation*}
which yields $E[\bL^-(\sigma)]\geq \underset{n\to\infty}{\Lim}E[\bL^-(\sigma_n)]$ as $\varepsilon>0$ was arbitrary. Conversely, the sub\-mar\-tin\-gale property of $\bL^-$  gives $E[\bL^-(\sigma)]\leq \underset{n\to\infty}{\Lim}E[\bL^-(\sigma_n)]$, where the limit exists by mono\-tonic\-ity. So we get $E[\bL^-(\sigma)]=\underset{n\to\infty}{\Lim}E[\bL^-(\sigma_n)]$, completing the proof of right-continuity.

2) Thanks to step 1), we can take as $L^\pm$ an RCLL version of $(\bL^\pm(t))_{0\leq t\leq T}$. To prove \eqref{D}, take $\sigma, \sigma_n\in\cS_{\tau,T}$ such that $\sigma_n\searrow \sigma$ and each $\sigma_n$ takes only finitely many values. Then \eqref{D} holds for each $\sigma_n$ and so
$\Lim_{n\to\infty}\bL^\pm(\sigma_n)=\Lim_{n\to\infty}L^\pm_{\sigma_n}=L^\pm_\sigma$
because $L^\pm$ are RCLL. Since all processes take values in $[0,1]$, dominated convergence yields 
$E[L^\pm_\sigma]=\Lim_{n\to\infty}E[\bL^\pm(\sigma_n)]=E[\bL^\pm(\sigma)]$
by the argument in step 1), and since the submartingale property, \eqref{D} for $\sigma_n$ and again dominated convergence give
$$
\bL^\pm(\sigma)
\leq
\Lim_{n\to\infty}E\big[ \bL^\pm(\sigma_n)\big|\F_\sigma\big]
=
\Lim_{n\to\infty}E\big[ L^\pm_{\sigma_n}\big|\F_\sigma \big]
=
L^\pm_\sigma,
$$
we obtain \eqref{D} for $\sigma$ as well. This proves part 2).

3) The equality in \eqref{F} follows directly from the definition \eqref{E}, \eqref{D} and the decomposition \eqref{B} in Proposition \ref{pJL}. The properties of the $\bJ$-family then immediately give the remaining assertion in part 2).
\ep
The next result gives an alternative description of the processes $L^\pm$ and some further useful properties.
\bl\label{mart}
Suppose that there exists a solution $\tvp^{(x,\tau)}$ to \eqref{A}. Then:

\vspace{0.25\baselineskip}
{\bf 1) }We have the decomposition
\begin{equation}
\tvp^{(x,\tau)} = x^+ \tvp^{(1,\tau)} + x^- \tvp^{(-1,\tau)}.
\label{strat-fact}
\end{equation}

\vspace{0.25\baselineskip}
{\bf 2) }For any $\sigma\in\cS_{\tau,T}$, we have on $\{V_\sigma(x,\tvp^{(x,\tau)})\gtrless0\}$ that
$$
L^\pm_\sigma
=
E\bigg[\bigg(1\pm\frac{\one_{\rrbracket \sigma,T\rrbracket }\tvp^{(x,\tau)}}{V_\sigma^\pm(x,\tvp^{(x,\tau)})}\sint S_T\bigg)^2\bigg|\F_\sigma\bigg]
=
E\bigg[1\pm\frac{\one_{\rrbracket \sigma,T\rrbracket }\tvp^{(x,\tau)}}{V_\sigma^\pm(x,\tvp^{(x,\tau)})}\sint S_T\bigg|\F_\sigma\bigg].
$$

\vspace{0.25\baselineskip}
{\bf 3) }The process $\, {}^\tau \widetilde M^{(x,\tau)} = \one_{\rrbracket \tau,T\rrbracket}\sint \widetilde M^{(x,\tau)}$ with
$$
\widetilde M^{(x,\tau)}:= (x+\tvp^{(x,\tau)}\sint S)^+L^+-(x+\tvp^{(x,\tau)}\sint S)^-L^-
$$
is a square-integrable martingale.

\vspace{0.25\baselineskip}

{\bf 4) }If $K:\Ombar\to 2^{\RR^d}\setminus\{\emptyset\}$ is convex-valued, then $(\vt\sint S)\widetilde M^{(x,\tau)}$ is a submartingale for all $\vt\in\Theta(K)$ with $\vt=0$ on $\llbracket 0,\tau\rrbracket$.
\el
\bp
1) The decomposition \eqref{strat-fact} of the optimal strategy is obtained like \eqref{eq-sep} directly from the fact that our optimisation problem is quadratic and the constraints are conic.

2) If there exists a solution $\tvp^{(x,\tau)}$ to \eqref{A}, we obtain by part 2) of Proposition \ref{prop:mop} that $J_\sigma(\tvp^{(x,\tau)};x,\tau)=E[|x+\tvp^{(x,\tau)}\sint S_T|^2|\F_\sigma]$ and therefore
$$
L^+_\sigma
=
E\bigg[\bigg(1+\frac{\one_{\rrbracket \sigma,T\rrbracket }\tvp^{(x,\tau)}}{V_\sigma^+(x,\tvp^{(x,\tau)})}\sint S_T\bigg)^2\bigg|\F_\sigma\bigg]
\quad\text{on $F:=\{V_\sigma(x,\tvp^{(x,\tau)})>0\}\in\F_\sigma$}
$$
by dividing in \eqref{B}. For the proof of the second equality, we can assume that the process $\vt := \frac{\one_{\rrbracket \sigma,T\rrbracket }\tvp^{(x,\tau)}}{V_\sigma^+(x,\tvp^{(x,\tau)})}\one_F$ is in $\TBK$ by part 2) of Lemma \ref{las} and by possibly shrinking $F$. Then the first equality implies for all $\varepsilon>-1$ that
\begin{align}
0
&\leq
\frac{E\big[ |1+((1+\varepsilon)\vt)\sint S_T|^2\big| \F_\sigma\big]-E\big[|1+\vt\sint S_T|^2\big| \F_\sigma\big]}{|\varepsilon|}
\nonumber\\
&=
-\sign(\varepsilon)E[(\vt\sint S_T)(1+\vt\sint S_T)| \F_\sigma]+|\varepsilon|E\big[ |\vt\sint S_T|^2 \big| \F_\sigma \big] .
\label{optimalitycondition}
\end{align}
Taking $\underset{\varepsilon\nearrow 0}{\varliminf}$ and $\underset{\varepsilon\searrow 0}{\varliminf}$ in \eqref{optimalitycondition} yields $E[(\vt\sint S_T)(1+\vt\sint S_T)| \F_\sigma]=0$,
which implies that $E[|1+\vt\sint S_T|^2| \F_\sigma]=E[1+\vt\sint S_T| \F_\sigma]$ and therefore the second asserted equality. The argument for $L^-_\sigma$ is completely analogous and therefore omitted.

3) Using the second equalities in part 2), we can write for $\sigma\in\cS_{\tau,T}$ that
$$
E[x+\tvp^{(x,\tau)}\sint S_T|\F_\sigma]
=
(x+\tvp^{(x,\tau)}\sint S_\sigma)^+L_\sigma^+-(x+\tvp^{(x,\tau)}\sint S_\sigma)^-L_\sigma^- ,
$$
which immediately gives that $\, {}^\tau \widetilde M^{(x,\tau)} = \one_{\rrbracket \tau,T\rrbracket}\sint \widetilde M^{(x,\tau)}$ is a square-integrable martingale.

4) Since $\vt\in\Theta(K)$ implies that $\one_{F\times(s,t]\cap\rrbracket \tau,T\rrbracket }\vt$ is in $\fK(0,\tau)$ for all $s\leq t$ and $A\in\F_s$, it follows from the first order condition of optimality for \eqref{A} that
\begin{align*}
&E\big[\one_F\big((\one_{\rrbracket \tau,T\rrbracket }\vt)\sint S_t-(\one_{\rrbracket \tau,T\rrbracket }\vt)\sint S_s\big)(x+\tvp^{(x,\tau)}\sint S_T)\big]
\\
&=E\big[ \big( (\one_{F\times(s,t]\cap\rrbracket \tau,T\rrbracket }\vt)\sint S_T\big) (x+\tvp^{(x,\tau)}\sint S_T) \big]\geq 0
\end{align*}
and therefore that $((\one_{\rrbracket \tau,T\rrbracket }\vartheta)\sint S_t) E[(x+\tvp^{(x,\tau)}\sint S_T)| \F_{t}]$, $0\leq t\leq T$, is a submartingale.
\ep
The martingale optimality principle in Proposition \ref{prop:mop} gives a dynamic description of the solution $\tvp = \tvp^{(x,0)}$ only for $J(\tvp;x,0)\ne0$. This can cause problems. But \eqref{E} shows that if $J(\tvp;x,0)$ becomes $0$, then either $V(x,\tvp)=0$ or $L^+=0$ or $L^-=0$. In the latter two cases, the payoffs $\one_{\{L^+_{\tau}=0\}}$ or $-\one_{\{L^-_{\tau}=0\}}$ with $\tau=\Inf\{t>0 \,|\, J_t(\tvp;x,0)=0\}\wedge T$ are in $G_T(\overline{\Theta(K\one_{\rrbracket \tau,T\rrbracket})})$, and in the terminology of Section 4 in \cite{S01}, these random variables provide approximate profits in $L^2$ which is a weak form of arbitrage. So intuitively, we have difficulties with describing $\tvp$ only if the basic model allows some kind of arbitrage. The next result, which generalises Lemma 3.10 in \cite{CK07}, gives a sufficient condition to prevent such problems.
\bl\label{lwac}
Suppose that there exist $N\in \MznlocP$ and $\ZN$ such that $(\E,\ZN)$ with $\E=\E(N)$ is regular and square-integrable and $S$ is an $\E$-local martingale. Then $L^\pm$ and their left limits $L^\pm_-$ are $(0,1]$-valued.
\el
\bp
We prove the assertion for $L^+$ and $L^+_-$ by way of contradiction; the completely analogous proof for $L^-$ and $L^-_-$ is omitted. Define $\tau:=\Inf\{t>0 \,|\, L^+_t=0\}\wedge T$ and suppose that $P[L^+_\tau=0]>0$. By \eqref{C}, \eqref{D} and the definition of $\tau$,
$$\underset{\vp\in\fK(0,\tau;\tau)}{\essinf}E\big[|1+\vp\sint S_T|^2\big|\F_\tau\big]\one_{\{L^+_\tau=0\}}=L^+_\tau\one_{\{L^+_\tau=0\}}=0
$$
and so there exists a sequence $(\vt^n)$ in $\fK(0,\tau;\tau)$ such that $((\vt^n\sint S_T)\one_{\{L^+_\tau=0\}})$ converges to $-\one_{\{L^+_\tau=0\}}$ in $ \LiiP$. Since $L^+_T=1$, we have that
$$\{L^+_\tau=0\}=\{L^+_\tau=0, \tau<T\}=\bigcup_{m=0}^\infty\{L^+_\tau=0, T_m\leq \tau<T_{m+1}\}$$
and hence $P[L^+_\tau=0, T_m\leq \tau<T_{m+1}]>0$ for some $m\in\N_0$. But each $\vt^n\sint S$ is an $(\E,\ZN)$-martingale by Corollary \ref{coremart}, and since $\ZNTm {}^{T_m}\E(N)$ is square-integrable, we get for every $F\in\F_\tau$ that
\begin{align*}
0
&=
\Lim_{n\to\infty}E\big[ \ZNTm {}^{T_m}\E(N)_T(\vt^n \sint S_T)\one_{\{L^+_\tau=0, T_m\leq \tau<T_{m+1}\}\cap F}\big]
\\
&=
-E\big[ \ZNTm {}^{T_m}\E(N)_\tau\one_{\{L^+_\tau=0, T_m\leq \tau<T_{m+1}\}\cap F}\big] .
\end{align*}
Since ${}^{T_m}\E(N)\ne0$ on $\llbracket T_m,T_{m+1}\llbracket$, choosing $F:= \{ {}^{T_m}\E(N)_\tau>0\}$ or $F:= \{ {}^{T_m}\E(N)_\tau<0\}$ gives a contradiction to the assumption that $P[L^+_\tau=0]>0$. So we get $L^+>0$.

To prove that $L^+_->0$, define the stopping time $\sigma:=\Inf\{t>0 \,|\, L^+_{t-}=0\}\wedge T$ and assume that 
$F_\infty:=\{L^+_{\sigma-}=0\}$ has $P[F_\infty]>0$. Because ${}^{T_m}\E(N)\ne0$ on $\llbracket T_m, T_{m+1}\llbracket$ and
$$
\{L^+_{\sigma-}=0\}=\{L^+_{\sigma-}=0, \sigma>0 \}=\bigcup_{m=0}^\infty\{L^+_{\sigma-}=0, T_m< \sigma\leq T_{m+1}\} ,
$$
there exists some $m\in\N_0$ with $P[F^{m,+}_\infty]>0$ or $P[F^{m,-}_\infty]>0$, where
$$
F^{m,\pm}_\infty := F_\infty\cap\{T_m< \sigma \leq T_{m+1}\} \cap \{ {}^{T_m}\E(N)_{\sigma-}\gtrless 0\}.
$$
We fix $m$ and treat without loss of generality the \lq\lq$+$\rq\rq\ case so that $P[F^{m,+}_\infty]>0$. Setting 
$\sigma_n:=\Inf\{t>0 \,|\, L^+_{t}\leq \frac{1}{n}\}\wedge T$ 
gives $\sigma_n<\sigma$ and $\sigma_n\nearrow\sigma$ $P$-a.s.~on $F_\infty$, and defining
$$
F^{m,+} _n:=\{0<L^+_{\sigma_n}\leq\frac{1}{n}\} \cap \{T_m\le\sigma_n<T_{m+1}\} \cap \{ {}^{T_m}\E(N)_{\sigma_n}>0\} \in\F_{\sigma_n}
$$
yields by the definition of $\sigma_n$ that 
$$
E\Big[\underset{\vp\in\fK(0,\sigma_n;\sigma_n)}{\essinf}E\big[ |1+\vp\sint S_T|^2\big|\F_{\sigma_n}\big] \one_{F^{m,+}_n}\Big]
=
E[L^+_{\sigma_n}\one_{F^{m,+}_n}]
\leq
\frac{1}{n}P[F^{m,+}_n].
$$
Thus there exist $\vp^n\in\fK(0,\sigma_n;\sigma_n)$ such that $\Lim_{n\to\infty}E\big[ |1+\vp^n\sint S_T|^2\one_{F^{m,+}_n}\big]=0$.
This implies as above via Corollary \ref{coremart} and the square-integrability of $\ZNTm {}^{T_m}\E(N)$ that
\begin{align*}
0
&=
\Lim_{n\to\infty}E\left[ \ZNTm {}^{T_m}\E(N)_T(\vp^n \sint S_T)\one_{F^{m,+}_n}\right]
=
-\Lim_{n\to\infty}E\left[ \ZNTm {}^{T_m}\E(N)_{\sigma_n}\one_{F^{m,+}_n}\right]
\\
&=
-E\big[ \ZNTm {}^{T_m}\E(N)_{\sigma-} \, \one_{F^{m,+}_\infty } \big].
\end{align*}
This contradicts the fact that $P[F^{m,+}_\infty]>0$ so that we must have $P[F_\infty]=0$.  
\ep
The lemma below allows us to parametrise the optimal strategy in terms of units of wealth. The proof uses the technique in \cite{DS96a}, which also appears in \cite{CKS98} and \cite{CK07}.
\bl\label{lpfw}
Suppose that $L^\pm$ and their left limits $L^\pm_-$ are $(0,1]$-valued and that there exists a solution $\tvp^{(x,\tau)}$ to \eqref{A}. Then there exists $\tpsi^{(x,\tau)}\in\cL(S)$ such that
\be
V(x,\tvp^{(x,\tau)})=x+\tvp^{(x,\tau)}\sint S=x\,\cE(\tpsi^{(x,\tau)}\sint S) \label{def:tpsi}
\ee
and
\be
L_t^\pm=E\big[ |\E(\tpsi^{(x,\tau)}\one_{\rrbracket t,T\rrbracket}\sint S)_T |^2\big|\F_t\big]
\quad\text{on $\{x+\tvp^{(x,\tau)}\sint S_t\gtrless0\}$.}\label{eq:Ltpsi}
\ee
\el
\bp
Define the stopping times $\sigma_n=\Inf\{t>0 \,|\, |V_t(x,\tvp^{(x,\tau)})|\leq \frac{|x|}{n+1}\}\wedge T$ for $n\in\N$, set $\sigma=\Lim_{n\to\infty}\sigma_n$ and $F=\bigcap_{n\in\N}\{\sigma_n <\sigma\}\in\bigvee_{n=1}^\infty\F_{\sigma_n}=\F_{\sigma-}$ and consider the square-integrable martingale $M^{(x,\tau)}_t=E[V_T(x,\tvp^{(x,\tau)})|\F_t]$ for $t\in[0,T]$. Lemma \ref{mart} yields
\begin{align}
M^{(x,\tau)}_t
&=
(x+\tvp^{(x,\tau)}\sint S_t)^+L^+_t-(x+\tvp^{(x,\tau)}\sint S_t)^-L^-_t
\quad\text{for $t\geq \tau$,}
\nonumber
\\
E\big[ (M^{(x,\tau)}_T)^2 \big| \F_t \big]
&=
\big((x+\tvp^{(x,\tau)}\sint S_t)^+\big)^2L^+_t+\big((x+\tvp^{(x,\tau)}\sint S_t)^-\big)^2L^-_t
\quad\text{for $t\geq \tau$,}
\label{eq:lpfw}
\end{align}
and since $L^\pm$ are $(0,1]$-valued and $\sigma_n\geq\tau$, we get $|M^{(x,\tau)}_{\sigma_n}| \leq \frac{|x|}{n+1}$, $|M^{(x,\tau)}_{\sigma_n}|>0$ on $\{\sigma_n<\sigma\}$, $F=\{M^{(x,\tau)}_{\sigma-}=0\}$ and $\one_FE[M^{(x,\tau)}_T|\F_{\sigma-}]=0$. Then the martingale property of $M^{(x,\tau)}$, conditioning on $\F_{\sigma-}$, and using Cauchy--Schwarz and \eqref{eq:lpfw} yields
\begin{align*}
\one_{\{\sigma_n <\sigma\}}
&=
E\bigg[\frac{M^{(x,\tau)}_T}{M^{(x,\tau)}_{\sigma_n}}\one_{\{\sigma_n <\sigma\}}\bigg|\F_{\sigma_n}\bigg]=E\bigg[\frac{M^{(x,\tau)}_T}{M^{(x,\tau)}_{\sigma_n}}\one_{\{\sigma_n <\sigma\}}\one_{F^c}\bigg|\F_{\sigma_n}\bigg]\\
&\leq
E\bigg[\bigg(\frac{M^{(x,\tau)}_T}{M^{(x,\tau)}_{\sigma_n}}\bigg)^2\one_{\{\sigma_n <\sigma\}}\bigg|\F_{\sigma_n}\bigg]^{\frac{1}{2}}P[F^c|\F_{\sigma_n}]^{\frac{1}{2}}\\
&\leq
\left(\frac{1}{L^+_{\sigma_n}}+\frac{1}{L^-_{\sigma_n}}\right)^{\frac{1}{2}}\one_{\{\sigma_n <\sigma\}}P[F^c|\F_{\sigma_n}]^{\frac{1}{2}}.
\end{align*}
Since
\begin{align*}
\one_F
&=
\Lim_{n\to\infty}\one_{\{\sigma_n <\sigma\}}\one_F\leq\Lim_{n\to\infty}\left(\frac{1}{L^+_{\sigma_n}}+\frac{1}{L^-_{\sigma_n}}\right)^{\frac{1}{2}}\one_F\one_{\{\sigma_n <\sigma\}}P[F^c|\F_{\sigma_n}]^{\frac{1}{2}}\\
&=
\left(\frac{1}{L^+_{\sigma-}}+\frac{1}{L^-_{\sigma-}}\right)^{\frac{1}{2}}\one_F\one_{F^c}=0,
\end{align*}
this gives $P[F]=0$ and therefore $V_-(x,\tvp^{(x,\tau)})\ne0$ on $\llbracket0,\sigma\rrbracket$ and $V(x,\tvp^{(x,\tau)})=0$ on $\llbracket\sigma,T\rrbracket$. Therefore $\tpsi^{(x,\tau)}:=\frac{\tvp^{(x,\tau)}}{V_-(x,\tvp^{(x,\tau)})}\one_{\llbracket0,\sigma\rrbracket}$ is well defined and satisfies \eqref{def:tpsi}. Plugging \eqref{def:tpsi} into the equations of part 2) of Lemma \ref{mart} yields \eqref{eq:Ltpsi} and completes the proof.
\ep
%%%%%%%%%%%%%%%%%%%%%%%%%%%%%%%%%%%%%%%%%%%%%%%%%%%%%%%%%%%%%%%%%%%%%%%%%%%%%%%%%%%%%%%%%
%%\newpage
\section{Local description and structure}\label{sec:ld}
In this section, we use the dynamic characterisation of the solution of \eqref{ap} to derive a local description for the structure of the optimal strategy. To that end, we first give a local description of the underlying processes by their differential semimartingale characteristics.

As in \cite{JS}, Theorem II.2.34, each $\R^d$-valued semimartingale $X$ has, with respect to some truncation function $h : \R^d\to\R^d$, the \emph{canonical representation}
$$X=X_0+X^c+ A^{X,h} +h(x)\ast(\mu^X-\nu^X)+ [x-h(x)]\ast\mu^X$$
with the jump measure $\mu^X$ of $X$ and its predictable compensator $\nu^X$. The quadruple $(b^X,c^X,F^X,B)$ of \emph{differential characteristics} of $X$ then consists of a predictable $\R^d$-valued process $b^X$, a predictable nonnegative-definite symmetric matrix-valued process $c^X$, a predictable process $F^X$ with values in the set of L\'evy measures on $\R^d$, and a predictable increasing RCLL process $B$ null at zero such that
$$
A^{X,h}= b^X\sint B, \qquad \la X^c\ra=c^X\sint  B,\qquad\nu^X=F^X\sint B.
$$
We use the same predictable process $B$ for all the finitely many semimartingales appearing in this paper, and since they are all special, we can and do always work with the (otherwise forbidden) truncation function $h(x)=x$, which simplifies computations considerably. We then write $A^X$ instead of $A^{X,h}$. For two (special) semimartingales $X$ and $Y$, we denote their joint differential characteristics by
$$(b^{X,Y},c^{X,Y}, F^{X,Y},B)=\left(\Big(\begin{array}{c}b^X\\
b^{Y}
\end{array}\Big),\Big(\begin{array}{cc}c^X& c^{XY}\\
c^{YX}& c^{Y}
\end{array}\Big), F^{X,Y},B\right).
$$
By adding $t$ to $B$, we can assume that $B$ is strictly increasing. Recall that $P_B=\PtB$. For the locally square-integrable semimartingale $S$, there exists by Proposition II.2.29 in \cite{JS} a predictable nonnegative-definite symmetric matrix-valued process $\widetilde c^M$ such that $\la M\ra = \widetilde c^M\sint B$, and it is given by
$
\widetilde c^M=c^S+\int xx^\T F^S(dx)-b^S(b^S)^\T \Delta B
$.

To prepare for the local description of the optimal strategy, we need some notation. For two $[0,1]$-valued (hence special) semimartingales $\ell^+$ and $\ell^-$, we look at their joint differential characteristics with $S$ and define the predictable functions 
\begin{align}
\mathfrak{g}^{1,\pm}(\psi):=\mathfrak{g}^{1,\pm}(\psi;S,\ell^+,\ell^-)
&:=
\ell^\pm_{-}\psi^\T c^S\psi\pm2\ell^\pm_{-}\psi^\T b^S\pm2 \psi^\T c^{S\ell^\pm},\label{g1+}\\
\mathfrak{g}^{2,\pm}(\psi):=\mathfrak{g}^{2,\pm}(\psi;S,\ell^+,\ell^-)
&:=
\ell^\pm_{-}\int\big(\big\{(1\pm\psi^\T u)^+\big\}^2-1\mp2\psi^\T u\big)F^S(du) \nonumber\\
&\phantom{:=\ }+\int\big(\big\{(1\pm\psi^\T u)^+\big\}^2-1\big)yF^{S,\ell^\pm}(du,dy)\nonumber\\
&\phantom{:=\ }+\int \big\{(1\pm\psi^\T u)^-\big\}^2(\ell^\mp_-+z) F^{S,\ell^\mp}(du,dz),\label{g2+}\\
\mathfrak{g}^\pm(\psi):=\mathfrak{g}^\pm(\psi;S,\ell^+,\ell^-)
&:=
\mathfrak{g}^{1,\pm}(\psi;S,\ell^+,\ell^-)+\mathfrak{g}^{2,\pm}(\psi;S,\ell^+,\ell^-).\label{g+}
\end{align}
All these functions have $\psi\in\R^d$ as arguments and depend on $\om,t$ via $\ell^\pm_{t-}(\om)$ and the joint characteristics of $S$ and $\ell^\pm$. For ease of notation, we shall drop in the proofs all superscripts ${}^\T$, writing $xy$ instead of $x^\T y$ for the scalar product of two vectors $x,y$.

Our first main result is now a local description of the optimal strategy $\tvp$ for \eqref{ap}. It is obtained by examining the drift rate of $J(\vt)$, as follows. Recall that the constraints are given by a predictable correspondence $K$ with closed cones as values.
\bt\label{mainthm1}
For each $\vt\in\TBK$, define a $K$-valued predictable process $\psi$ via
\be
\psi := \one_{\{V_{-}(x,\vt)\ne0\}} \frac{\vt}{|V_{-}(x,\vt)|}+\one_{\{V_{-}(x,\vt)=0\}}\vt\label{defpsi}
\ee
or equivalently
$$
\vt=:V_-^+(x,\vt)\psi+V_-^-(x,\vt)\psi+\one_{\{V_{-}(x,\vt)=0\}}\psi.
$$
Then:

\vspace{0.25\baselineskip}
{\bf{1)} }The finite variation part of $J(\vt)$ is given by $A(\vt)=b^{J(\vt)}\sint B$ with
\begin{align*}
b^{J(\vt)}
&=
\big(V_-^+(x,\vt)\big)^2\big\{\fg^+(\psi;S,\ell^+,\ell^-)+b^{\ell^+}\big\}+\big(V_-^-(x,\vt)\big)^2\big\{\fg^-(\psi;S,\ell^+,\ell^-)+b^{\ell^-}\big\}\nonumber\\
&\phantom{=\ }
+\one_{\{V_{-}(x,\vt)=0\}}\Big(\int\big((\psi^\T u)^+\big)^2(\ell^+_{-}+y)F^{S,\ell^+}(du,dy)+\ell^-_-\psi^\T c^S\psi\nonumber\\
&\phantom{=\ }
+\int\big((\psi^\T u)^-\big)^2(\ell^-_{-}+z)F^{S,\ell^-}(du,dz)\Big)\geq 0.
\end{align*}

\vspace{0.25\baselineskip}
{\bf{2)} }If there exists a solution $\tvp=\tvp^{(x,0)}\in\TBK$ to problem \eqref{ap} with the property that
$$
V(x,\tvp)=x+\tvp\sint S=x\,\E(\tpsi\sint S),
$$
then the joint differential characteristics of $(S,L^+,L^-)$ satisfy the two coupled equations
\begin{equation}
b^{L^\pm}
=
-\min_{\psi\in K}\mathfrak{g}^\pm(\psi;S,L^+,L^-)=-\mathfrak{g}^\pm(\pm\tpsi;S,L^+,L^-)
\quad\text{on $\{V_-(x,\tvp) \gtrless 0\}$.}
\label{mainthm1:eq:1}
\end{equation}
\et
\bp
1) Since $J(\vt)$ is given by \eqref{E}, finding its drift rate $b^{J(\vt)}$ is a straightforward, but lengthy computation; this is done in Lemma \ref{lcdc} below. Then $b^{J(\vt)}$ is nonnegative because $J(\vt)$ is a submartingale by the martingale optimality principle in Proposition \ref{prop:mop}.

2) The basic idea to prove the first equality is (as usual) to assume that the set
$$D:=\big\{\omt \,\big|\, b^{L^+}>-\min_{\psi\in K}\mathfrak{g}^+(\psi;S,L^+,L^-)\big\}\cap\{x\,\E(\tpsi\sint S)_->0\}$$
has $P_B(D)>0$ and then to construct from $D$ via measurable selection a strategy $\vt$ in $\TBK$ which violates the submartingale property of $J(\vt)$. This simple idea is technically a bit involved because one must ensure that $\vt$ is $K$-admissible and that there exists a set $D'\in\cP$ with $D' \subseteq D$, $P_B(D')>0$ and $V_-(x,\vt)>0$ on $D'$. The details are as follows.

Since $V(x,\tvp)=x\,\E(\tpsi\sint S)$ is a stochastic exponential, it changes sign only at jumps with $\tpsi \Delta S<-1$, which $P$-a.s.~can only happen a finite number of times. So there exist stopping times $\tau_1\leq\tau_2$ such that $P_B(D\,\cap\,\rrbracket \tau_{1},\tau_2\rrbracket)>0$ and $x\,\E(\tpsi\sint S)_->0$ on $\rrbracket \tau_{1},\tau_2\rrbracket$. By part 2) of Lemma \ref{las}, we can choose $F_{\varepsilon}\in\F_{\tau_1}$ such that $\tvp\one_{\llbracket 0,\sigma_1\rrbracket }\in\TBK$ and $(x+\tvp\sint S_{\sigma_1})\one_{F_{\varepsilon}}\geq0$ is uniformly bounded and $D_{\ve}:=D\,\cap\,\rrbracket \sigma_{1},\sigma_2\rrbracket$ has $P_B(D_{\ve})>0$, where $\sigma_i:=\tau_i\one_{ F_{\varepsilon}}+T\one_{ F_{\varepsilon}^c}$ for $i=1,2$ are stopping times. Because $\fg^+$ is a Carath\'eodory function by Lemma \ref{lg} below and $K$ is a predictable correspondence, we can construct by Propositions \ref{propCara} and \ref{Castaing} a $K$-valued predictable process $\vp$ with $\fg^+(\vp)<-b^{L^+}$ on $D_\ve$ and $\fg^+(\vp)=0$ else. After possibly shrinking $D_\ve$, we can also assume without loss of generality that $\vp$ is bounded, which implies that $\vp$ is in $\cL(S)$ so that $\vp\sint S$ is well defined and has $P$-a.s.~only a finite number of jumps with $\vp\Delta S<-1$. Thus there exists stopping times $\vr_1\leq\vr_2$ such that $D':=D_\ve\cap\,\rrbracket \vr_{1},\vr_2\rrbracket$ has $P_B(D')>0$ and $\E(\psi\sint S)_->0$ on $\rrbracket \vr_{1},\vr_2\rrbracket$, where $\psi:=\vp\one_{\rrbracket \vr_{1},\vr_2\rrbracket}$. By stopping $\E(\psi\sint S)_-$ and $S$, we can even choose $\vr_2$ such that $\E(\psi\sint S)_-$ is bounded and $\E(\psi\sint S)_-\psi\in\Theta(K)$; this uses that $K$ is cone-valued. Moreover, since $(x+\tvp\sint S_{\sigma_1})\one_{F_{\varepsilon}}$ is bounded, also $(x+\tvp\sint S_{\sigma_1})\one_{F_{\varepsilon}}\E(\psi\sint S)_-\psi$ is in $\Theta(K)$. Therefore the sum
$$\vt:=\tvp\one_{\llbracket 0,\sigma_1\rrbracket }+(x+\tvp\sint S_{\sigma_1})\one_{F_{\varepsilon}}\E(\psi\sint S)_-\psi$$ is in $\TBK$ and has $(x+\vt\sint S)_->0$ and $\fg^+(\frac{\vt}{(x+\vt\ssint S)_-})=\fg^+(\psi)<-b^{L^+}$ on $D'$. In view of part 1), $\one_{D'}\sint A(\vt)=(\one_{D'}b^{J(\vt)})\sint B=(\one_{D'}(x+\vt\sint S)_-\{\fg^+(\psi)+b^{L^+}\})\sint B$ is strictly decreasing on a non-negligible set, and so $J(\vt)$ cannot be a submartingale. This contradicts the martingale optimality principle and thus establishes the equality for $b^{L^+}$. The argument for $b^{L^-}$ is completely analogous and therefore omitted.
\ep
To explain the significance as well as the limitations of Theorem \ref{mainthm1}, let us suppose that we have an optimal strategy $\tvp$ for problem \eqref{ap}. Then part 2) of Theorem \ref{mainthm1} gives a kind of BSDE description for the pair $(L^+,L^-)$ since it expresses their drift rates in terms of their joint semimartingale characteristics with $S$. However, this description is not yet fully informative on its own. A closer look at \eqref{mainthm1:eq:1} shows that we only have a description of the drift of $L^+$ (or $L^-$) when $V_-(x,\tvp)$ is positive (or negative). Once $V(x,\tvp)$ hits $0$, it stays there, being a stochastic exponential, and we can no longer tell how $L^\pm$ behave. Even worse, $V(x,\tvp)$ might jump across $0$ so that we immediately lose track of the drift of $L^+$ or $L^-$, depending on whether the jump goes downwards or upwards. To overcome this difficulty and obtain a full characterisation of $L^\pm$, we must be able to ``restart $V(x,\tvp)$ whenever it jumps across or to $0$''. This can be achieved by assuming that not only \eqref{ap}, but each problem \eqref{A} for $x$ and $\tau$ has a solution. This key insight can be traced back to \v Cern\'y and Kallsen \cite{CK07}.

The second condition we need to get a description of $L^\pm$ is that these processes as well as their left limits are strictly positive. As already explained before Lemma \ref{lwac}, this can be interpreted as a kind of absence-of-arbitrage condition. In fact, if -- as in \cite{CK07} -- there exists an equivalent local martingale measure for $S$ with density in $\LiiP$, that condition is automatically satisfied; a slightly more general result is given in Lemma \ref{lwac} above. For the case without constraints, we provide a sharp result in Theorem \ref{T6.2} below.
\begin{cor}\label{corld}
Suppose that $L^\pm$ and their left limits $L^\pm_-$ are all $(0,1]$-valued and that there exists a solution $\tvp^{(x,\tau)}$ to \eqref{A} for any $x\in\R$ and any stopping time $\tau$. Then the joint differential characteristics of $(S,L^+,L^-)$ satisfy
\be
b^{L^+}=-\min_{\psi\in K}\mathfrak{g}^+(\psi;S,L^+,L^-)
\qquad\text{and}\qquad
b^{L^-}=-\min_{\psi\in K}\mathfrak{g}^-(\psi;S,L^+,L^-).\label{eq:corld:a}
\ee
Moreover, for all $x\in\R$ and all stopping times $\tau$, there exists a solution to the SDE
\be
d V_t^{(x,\tau)}=\big( (V_{t-}^{(x,\tau)})^+\widetilde\psi_t^++(V_{t-}^{(x,\tau)})^-\widetilde\psi_t^-\big)\one_{\rrbracket \tau,T\rrbracket }\,dS_t,\quad V_0^{(x,\tau)}=V_\tau^{(x,\tau)}=x \label{eq:corld:b}
\ee
with $\widetilde\psi^{\pm}\in\underset{\psi\in K}{\argmin}\mathfrak{g}^\pm(\psi;S,L^+,L^-)$ on 
$\{V^{(x,\tau)}_- \gtrless0\}\,\cap\,\rrbracket\tau,T\rrbracket$
and $\tpsi^\pm \one_{\{V^{(x,\tau)}_- \gtrless0\}\cap\rrbracket\tau,T\rrbracket}$ in $\cL(S)$, and we have 
\be
\tvp^{(x,\tau)} =\big(( V_{-}^{(x,\tau)})^+\widetilde\psi^{+}+( V_{-}^{(x,\tau)})^-\widetilde\psi^{-}\big)\one_{\rrbracket \tau,T\rrbracket }.\label{eq:corld:c}
\ee
Note that $\widetilde\psi^\pm$ are not the positive and negative parts of the process $\widetilde\psi$ from Theorem \ref{mainthm1}.
\end{cor}
\bp
By Lemma \ref{lpfw}, we have $V(x,\tvp^{(x,\tau)})=x\,\E(\tpsi^{(x,\tau)}\sint S)$ for some $\tpsi^{(x,\tau)}\in\cL(S)$ with $\tpsi^{(x,\tau)}=\tpsi^{(x,\tau)}\one_{\rrbracket \tau,T\rrbracket }$ so that $\tpsi^\pm:=\tpsi^{(x,\tau)}\one_{\{V_-(x,\tvp^{(x,\tau)}) \gtrless 0\}}$ are in $\cL(S)$ and yield \eqref{eq:corld:b} with $V^{(x,\tau)}:=V(x,\tvp^{(x,\tau)})$. Moreover, \eqref{mainthm1:eq:1} in Theorem \ref{mainthm1} shows that $\tpsi^\pm$ are minimisers for $\fg^\pm$ on $\{V_-(x,\tvp^{(x,\tau)})\gtrless0\}\,\cap\,\rrbracket\tau,T\rrbracket$, and finally \eqref{eq:corld:c} holds by construction because $V^{(x,\tau)}=V(x,\tvp^{(x,\tau)})=x+\tvp^{(x,\tau)}\sint S$.
\ep
\begin{remark}
For the purpose of \emph{constructing} an optimal strategy, the result in Corollary \ref{corld} is not yet optimal. Ideally, one would like to take any minimisers $\tpsi^\pm$ for $\fg^\pm$, solve the SDE \eqref{eq:corld:b} and obtain that $\tvp^{(x,\tau)}$ \emph{defined} by \eqref{eq:corld:c} is optimal. However, it is not obvious whether these $\tpsi^\pm$ are automatically in $\cL(S)$. (That would of course imply solvability of \eqref{eq:corld:b}, and even optimality of $\tvp^{(x,\tau)}$ if that strategy is $K$-admissible.)
\end{remark}
Before we proceed with our BSDE descriptions, let us briefly return to the classical (but constrained) Markowitz problem in \eqref{CMP}. For given initial wealth $x$ and target mean $m$, we know from Lemma \ref{lgds} that the optimal strategy is given by
$
\textstyle
\tvt^{(m,x)}=\frac{m-x}{E[\widetilde\vp\ssint S_T]}\widetilde\vp,
$
where $\widetilde\vp = \widetilde\vp^{(-1,0)}$ solves \eqref{A} for $x=-1, \tau=0$. To express $\tvt^{(m,x)}$ in feedback form, write
\be
% \textstyle
V(x, \tvt^{(x,m)})
=
x + \frac{m-x}{E[\widetilde\vp\sint S_T]} \big( V(-1,\widetilde\vp)+1 \big)
=
\widetilde m + \frac{m-x}{E[\widetilde\vp\sint S_T]}V(-1, \widetilde\vp) 
\label{auxeq1}
\ee
with
$$
% \textstyle
\widetilde m
:=
x + \frac{m-x}{E[\widetilde\vp\sint S_T]}
=
\frac{m-xE[1-\widetilde\vp\ssint S_T]}{E[\widetilde\vp\sint S_T]}.
$$
By Corollary \ref{corld}, we have
$
\widetilde\vp^{(-1,0)} = (V_-^{(-1,0)})^+ \widetilde\psi^+ + (V_-^{(-1,0)})^- \widetilde\psi^-
$
and therefore
$$
\tvt^{(m,x)}
=
\big( V_-(x, \tvt^{(m,x)}) - \widetilde m \big)^+ \widetilde \psi^+ + \big( V_-(x, \tvt^{(m,x)}) - \widetilde m \big)^- \widetilde \psi^-
$$
by plugging in for $V^{(-1,0)} = V(-1, \widetilde\vp)$ from \eqref{auxeq1}. This shows that $\tvt^{(m,x)}$ is indeed a \emph{state feedback control}, and it also makes it clear that the critical level for switching between the \lq\lq positive and negative case strategies\rq\rq\ $\widetilde\psi^+$ and $\widetilde\psi^-$ is not zero (as one might think from the appearance of positive and negative parts), but rather $\widetilde m$.

Having found in Theorem \ref{mainthm1} and Corollary \ref{corld} necessary conditions for optimality, we now turn to sufficient ones.
\bt[Verification theorem]\label{vt1}
Let $\ell^\pm$ be semimartingales such that
\bi
\item[\bf{1)}]$\ell^\pm$ and their left limits $\ell^\pm_{-}$ are all $(0,1]$-valued and $\ell^\pm_T=1$.
\item[\bf{2)}]The joint differential characteristics of $(S,\ell^+,\ell^-)$ satisfy
\be
b^{\ell^+}=-\min_{\psi\in K}\mathfrak{g}^+(\psi;S,\ell^+,\ell^-)
\qquad\text{and}\qquad b^{\ell^-}=-\min_{\psi\in K}\mathfrak{g}^-(\psi;S,\ell^+,\ell^-).\label{eq:vt1:a}
\ee
\item[\bf{3)}]The solution to the SDE
\begin{equation}
d V_t=( V_{t-}^+\widetilde\psi_t^++V_{t-}^-\widetilde\psi_t^-)\,dS_t,\quad V_0=x\label{eq:vt1:b}
\end{equation}
with $\widetilde\psi^{\pm}\in\underset{\psi\in K}{\argmin}\mathfrak{g}^\pm(\psi)$ on $\{V_-\gtrless0\}$ exists and satisfies that
\be
\bar\vp:=V_{-}^+\widetilde\psi^{+}+ V_{-}^-\widetilde\psi^{-}\in\TBK .
\label{eq:vt1:c}
\ee
\ei
Then $\widetilde\vp:=\bar\vp$ is the solution to \eqref{ap}. In particular, $(V^+)^2\ell^++( V^-)^2\ell^-$ is of class (D).
\et
To better explain the significance of our results, let us rewrite the drift descriptions \eqref{eq:corld:a} and \eqref{eq:vt1:a} into a BSDE as follows. Consider the pair of coupled backward equations
\begin{equation}
\ell^\pm
=
-\inf\limits_{\psi\in K}\fg^\pm(\psi;S,\ell^+,\ell^-)\sint B+H^{\ell^\pm}\sint S^c+W^{\ell^\pm}\ast(\mu^S-\nu^S)+N^{\ell^\pm},\quad \ell^\pm_T=1,\label{BSDE:1}
% \\
% \ell^-
% &=
% -\Inf_{\psi\in K}\fg^-(\psi;S,\ell^+,\ell^-)\sint B+H^{-}\sint S^c+W^{-}\ast(\mu^S-\nu^S)+N^{-},\quad \ell^-_T=1,
% \label{BSDE:2}
\end{equation}
where a solution is a tuple $(\ell^\pm,H^{\ell^\pm},W^{\ell^\pm},N^{\ell^\pm})$ satisfying suitable properties; see below for a more precise formulation. Then Corollary \ref{corld} says that the opportunity processes $L^\pm$ from \eqref{D} satisfy the BSDE system \eqref{BSDE:1}, and Theorem \ref{vt1} conversely allows us to construct from a solution to \eqref{BSDE:1} a solution to the basic problem \eqref{ap}, if the natural candidate strategy $\bar\vp$ from \eqref{eq:vt1:c} has sufficiently good properties.
\begin{remark}\label{rm:vt1}
More generally, we could use Theorem \ref{vt1} to construct solutions to \eqref{A} for any $x\in\R$ and stopping time $\tau$. Indeed, if we replace the SDE \eqref{eq:vt1:b} with \eqref{eq:corld:b}, the definition of $\bar\vp$ in \eqref{eq:vt1:c} by \eqref{eq:corld:c} and assume that $\bar\vp^{(x,\tau)}$ is in $\TBK$, then $\bar\vp^{(x,\tau)}$ is the solution to \eqref{A}. The argument is exactly the same as below for problem \eqref{ap}.
\end{remark}
\bp[Proof of Theorem \ref{vt1}]
For $\vt\in\TBK$, define $j(\vt)=(V^+(x,\vt))^2\ell^++(V^-(x,\vt))^2\ell^-$ and a $K$-valued predictable process $\psi$ by \eqref{defpsi} so that $\vt=V_-^+(x,\vt)\psi+V_-^-(x,\vt)\psi+\one_{\{V_{-}(x,\vt)=0\}}\psi$. If $\vt\in\TS(K)$, then $\sup_{0\leq t\leq T}|V_t(x,\vt)|\in\LiiP$. Since $\ell^\pm$ are $(0,1]$-valued, we then have $\sup_{0\leq t\leq T}|j_t(\vt)|\in L^1(P)$ and so $j(\vt)$ is a special semimartingale with canonical decomposition $j(\vt)=j_0(\vt)+M^{j(\vt)}+A^{j(\vt)}$. Lemma \ref{lcdc} below gives $A^{j(\vt)}=b^{j(\vt)}\sint B$ with
\begin{align*}
b^{j(\vt)} = \bar b^{\vt}
&=
\big(V_-^+(x,\vt)\big)^2\big\{\fg^+(\psi;S,\ell^+,\ell^-)+b^{\ell^+}\big\}+\big(V_-^-(x,\vt)\big)^2\big\{\fg^-(\psi;S,\ell^+,\ell^-)+b^{\ell^-}\big\}\nonumber\\
&\phantom{=\ }
+\one_{\{V_{-}(x,\vt)=0\}}\Big( \int\big((\psi u)^+\big)^2(\ell^+_{-}+y)F^{S,\ell^+}(du,dy)+\ell^-_-\psi c^S\psi\nonumber\\
&\phantom{=\ }
+\int\big((\psi u)^-\big)^2(\ell^-_{-}+z)F^{S,\ell^-}(du,dz)\Big).
\end{align*}
Since $\bar b^{\vt}\geq 0$ by the BSDE \eqref{eq:vt1:a} in 2) and because $\ell^\pm$ are nonnegative, $j(\vt)$ is therefore a submartingale, and using $|V_T(x,\vt)|^2 = j_T(\vt)$ due to $\ell^\pm_T = 1$ gives
\be
E\big[|V_T(x,\vt)|^2\big]\geq E\big[(x^+)^2\ell_0^++(x^-)^2\ell_0^-\big].\label{eq:pr:vt}
\ee
Because $\vt\in\TS(K)$ was arbitrary and the closure in $L^2$ of $G_T(\TS(K))$ contains $G_T(\TBK)$, by definition, \eqref{eq:pr:vt} extends to all $\vt\in\TBK$.

To show that $\bar\vp$ is optimal, we want to argue that $j(\bar\vp)$ is a supermartingale, since we then get the reverse inequality in \eqref{eq:pr:vt} which is enough to conclude. Because $\bar\vp$ is only in $\TBK$, however, we do not know a priori if $j(\bar\vp)$ is special and thus must localise as in Lemma \ref{lcdc}. So we define for each $n\in\N$ the set $D_n:=\{|\bar\vp|\leq n\}\in\cP$ and $X^n:=\one_{D_n}\sint j(\bar\vp)=j^n(\bar\vp)$. We first note that \eqref{eq:vt1:c} and \eqref{eq:vt1:b} imply that $V=V(x,\bar\vp)$. The SDE \eqref{eq:vt1:b} then implies that $V$ remains at $0$ after $V_-$ hits zero, and so $\bar\vp\one_{\{V_-=0\}}=0$ by \eqref{eq:vt1:c}. For $\bar\psi$ defined from $\bar\vp$ via \eqref{defpsi} or \eqref{defpsi2} in Lemma \ref{lcdc} below, we then get
$$
\bar\psi=\bar\vp=0 
\quad\text{on $\{V_-=0\}=\{V_-(x,\bar\vp)=0\}$}$$
and therefore from \eqref{defb2} below that
$$\bar b^{\bar\vp}=\big(V_-^+(x,\bar\vp)\big)^2\big\{\fg^+(\bar\psi;S,\ell^+,\ell^-)+b^{\ell^+}\big\}+\big(V_-^-(x,\bar\vp)\big)^2\big\{\fg^-(\bar\psi;S,\ell^+,\ell^-)+b^{\ell^-}\big\}.
$$
But \eqref{defpsi} also gives that $\bar\vp=V_-^+(x,\bar\vp)\bar\psi+V_-^-(x,\bar\vp)\bar\psi=V_-^+\bar\psi+V_-^-\bar\psi$, and comparing this to \eqref{eq:vt1:c} shows that
$$
\bar\psi=\widetilde \psi^+ \text{ on $\{V_->0\}$}
\qquad\text{and}\qquad
\bar\psi=\widetilde \psi^- \text{ on $\{V_-<0\}$.}
$$
Because $\widetilde\psi^\pm$ are minimisers for $\fg^\pm$, we obtain that $\bar b^{\bar\vp}\equiv0$.

Now each $X^n$ is by Lemma \ref{lcdc} below and the above argument a special semimartingale with finite variation part $A^{X^n}=A^{j^n(\bar\vp)}=b^{j^n(\bar\vp)}\sint B=(\one_{D_n}\bar b^{\bar\vp})\sint B\equiv 0$. So each $X^n$ is a local martingale, which means that $j(\bar\vp)$ is a $\sigma$-martingale. Since $j(\bar\vp)\ge0$, it is therefore a supermartingale and so $\bar\vp$ solves \eqref{ap}. By part 2) of Proposition \ref{prop:mop}, $j(\bar\vp)$ is then even a martingale on $[0,T]$ and hence in particular of class (D).
\ep
We now return to the formulation of the equations \eqref{eq:corld:a} or \eqref{eq:vt1:a} as a coupled system of BSDEs. We first recall that by Proposition II.2.29 and Lemma III.4.24 in \cite{JS}, any special semimartingale $\ell$ can be decomposed as
\be
\ell
=
A^{\ell}+H^{\ell}\sint S^c+W^{\ell}\ast(\mu^S-\nu^S)+N^{\ell}\label{delta}
\ee
with $H^{\ell}\in\Lzloc(S^c)$, $W^{\ell}\in G_{\rm loc}(\mu)$ and $N^{\ell}\in\MnlocP$ such that $\la S^c,(N^{\ell})^c\ra=0$ and $M^P_\mu(\Delta N^{\ell}|\widetilde \cP)=0$. Then
$$\Delta \ell=\Delta A^{\ell}+(W^{\ell}-\Whell)\one_{\{\Delta S\ne0\}}+\Delta N^{\ell}$$
and therefore
\be
{}^\mathbf{p}(\Delta \ell\Delta S)=\int\big(\Delta A^{\ell}+(W^{\ell}(u)-\Whell)\big)uF^S(du).\label{alpha}
\ee
This allows us to rewrite the functions $\fg^\pm$ from \eqref{g1+}--\eqref{g+} as
\begin{align}
\mathfrak{g}^\pm(\psi;S,\ell^+,\ell^-)
&=
\ell^\pm_{-}\psi^\T c^S\psi\pm2\ell^\pm_{-}\psi^\T b^S\pm2 \psi^\T c^{S}H^{\ell^\pm}\nonumber\\
&\phantom{=\ }
+\ell^\pm_{-}\int\big(\big\{(1\pm\psi^\T u)^+\big\}^2-1\mp2\psi^\T u\big)F^S(du)\nonumber\\
&\phantom{=\ }
+\int\big(\big\{(1\pm\psi^\T u)^+\big\}^2-1\big) \big( \Delta A^{\ell^\pm}+W^{\ell^\pm}(u)-\Whellpm \big) F^{S}(du)\nonumber\\
&\phantom{=\ }
+\int \big\{(1+\psi^\T u)^-\big\}^2 \big(\ell^\mp_-+\Delta A^{\ell^\mp}+W^{\ell^\mp}(u)-\Whellmp \big) F^{S}(du)\nonumber\\
&=:
\mathfrak{h}^\pm(\psi;S,\ell^+,\ell^-).\label{beta}
\end{align}
We now consider the coupled system of backward equations
\be
\ell^\pm=-\Inf_{\psi\in K}\fh^\pm(\psi;S,\ell^+,\ell^-)\sint B+H^{\ell^\pm}\sint S^c+W^{\ell^\pm}\ast(\mu^S-\nu^S)+N^{\ell^\pm},\quad \ell^\pm_T=1.\label{BSDE2}
\ee
A \emph{solution} of \eqref{BSDE2} consists of tuples $(\ell^\pm,H^{\ell^\pm},W^{\ell^\pm},N^{\ell^\pm})$ such that $H^{\ell^\pm}$ are in $\Lzloc(S^c)$, $W^{\ell^\pm}$ are in $G_{\rm loc}(\mu)$, $N^{\ell^\pm}$ are in $\MnlocP$ with $\la S^c,(N^{\ell^\pm})^c\ra=0$ and $M^P_\mu(\Delta N^{\ell^\pm}|\widetilde \cP)=0$, and $\ell^\pm$ are (special) semimartingales with values in $[0,1]$. Moreover, being a solution also includes the condition that $\Inf_{\psi\in K}\fh^\pm(\psi;S,\ell^+,\ell^-)$ are finite-valued processes. For brevity, we sometimes call only $(\ell^+,\ell^-)$ a solution. Then Corollary \ref{corld} can be restated as
\begin{cor}\label{cor:BSDE}
Suppose that $L^\pm$ and their left limits $L^\pm_-$ are all $(0,1]$-valued and that there exists a solution to \eqref{A} for any $x\in\R$ and any stopping time $\tau$. Then the opportunity processes satisfy the coupled BSDE system
\begin{equation}
% \begin{array}{ll}
L^\pm
=
-\Inf_{\psi\in K}\fh^\pm(\psi;S,L^+,L^-)\sint B+H^{L^\pm}\sint S^c+W^{L^\pm}\ast(\mu^S-\nu^S)+N^{L^\pm},\quad L^\pm_T=1.
% \\
% L^-
% &=
% -\Inf_{\psi\in K}\fh^-(\psi;S,L^+,L^-)\sint B+H^{L^-}\sint S^c+W^{L^-}\ast(\mu^S-\nu^S)+N^{L^-},\quad L^-_T=1.
% \end{array}
\label{L-BSDE}
\end{equation}
Moreover, 
% $(L^\pm,H^{L^\pm},W^{L^\pm},N^{L^\pm})$ is a solution of \eqref{BSDE2} which is nice in the sense that
%
there exist $K$-valued processes $\tpsi^\pm$ such that
$$\fh^\pm( \tpsi^\pm;S,L^+,L^-)=\Inf_{\psi\in K}\fh^\pm(\psi;S,L^+,L^-).
$$
\end{cor}
The result in Corollary \ref{cor:BSDE} can be viewed as giving existence of a solution to the BSDE system \eqref{BSDE2}, and so it is natural to ask about uniqueness. For the case of an It\^o process $S$ in a Brownian filtration, Hu and Zhou \cite{HZ04} obtain a uniqueness result in the class of those solutions which have both $\ell^\pm$ uniformly bounded away from 0. However, this also rests on very restrictive assumptions on the It\^o coefficients of $S$ (uniformly bounded drift and uniformly elliptic volatility matrix), and one should not expect to have uniqueness in general. In fact, one can deduce from Example 3.26 in \cite{CK07} and the counterexample in \cite{CK08} that the opportunity processes $L^\pm$ are not the only solution to the BSDE system \eqref{BSDE2}, not even in the unconstrained case and if $S$ is continuous and under uniform integrability assumptions. Nevertheless, there is a positive result, motivated by similar ones in \cite{N09}: It turns out that $L^\pm$ are the \emph{maximal} processes which satisfy \eqref{BSDE2}.
\bl\label{maximal}
The opportunity processes $L^\pm$ satisfy $L^\pm\geq \ell^\pm$ for any solution $(\ell^+,\ell^-)$ of the BSDE \eqref{BSDE2}. In particular, under the assumptions of Corollary \ref{corld}, $(L^+,L^-)$ is the maximal solution of \eqref{BSDE2}.
\el 
\bp
This argument only uses the definitions of $L^\pm$ in \eqref{D} and \eqref{C} as essential infima. Let $(\ell^+,\ell^-)$ be any solution to \eqref{BSDE2} and define $\tau:=\Inf\{t>0 \,|\, \ell^+_t>L^+_t\}\wedge T$. By \eqref{D}, there exists a sequence $(\vt^n)$ in $\TS(K\one_{\rrbracket\tau,T\rrbracket})$ such that $\Lim_{n\to\infty}E[|V_T(1,\vt^n)|^2|\F_\tau]=L^+_\tau$ $P$-a.s.~The same argument as in the proof of Lemma \ref{lcdc} then shows that the process \hbox{$j(\vt^n)=(V^+(1,\vt^n))^2\ell^++(V^-(1,\vt^n))^2\ell^-$} is a submartingale, and so we obtain from $\ell^+_T=1$ and $V_\tau(1,\vt^n)=1$ that
$$\ell^+_\tau\leq \Lim_{n\to\infty}E\big[|V_T(1,\vt^n)|^2\big|\F_\tau\big]=L^+_\tau.
$$
By the definition of $\tau$, this implies that $P[\tau<T]=0$ and therefore that $L^+\geq \ell^+$ $P$-a.s. The proof of $L^-\geq\ell^-$ $P$-a.s.~is analogous and therefore omitted.
\ep
\begin{remark}
Due to the coupling term coming from $\mathfrak{h}^\pm$, the BSDE system \eqref{BSDE2} is very complicated. It has a nonlinear non-Lipschitz generator plus a driver with jumps, so that finding a solution by general BSDE techniques seems a formidable challenge. It is fortunate (and inherent to our approach) that we do not need to tackle this issue. We exploit instead that \eqref{L-BSDE} is intimately related to a stochastic control problem and prove directly existence of a solution to the latter, which then yields existence of a solution to \eqref{L-BSDE}. In that sense, we use BSDEs not for their own sake, but only as a tool to describe the value process of our stochastic control problem.
\end{remark}
%
%%%%%%%%%%%%%%%%%%%%%%%%%%%%%%%%%%%%%%%%%%%%%%%%%%%%%%%%%%%%%%%%%%%%%%%%%%%%%%%%%%%%%%%%%%%%%%%%%%%%%%%%%%%%%%%%%%%%%%%%%%%%%%%%%%%%%%%%%%%%%%%%%%%%%%
\section{Proofs}\label{sec:pr}
This section contains the more technical proofs. Several results and computations do not use the precise definition \eqref{D} of the processes $L^\pm$, but only some of their properties. To emphasise this, we formulate the corresponding results here for generic processes $\ell^\pm$. Recall that we drop the superscript ${}^\T$ in all proofs.

We first show that the predictable functions in \eqref{g1+}--\eqref{g+} are well defined and have nice properties.
\bl\label{lg}
Let $\ell^\pm$ be two $[0,1]$-valued semimartingales. Then the predictable functions $\fg^{1,\pm}$, $\fg^{2,\pm}$ and $\fg^\pm$ defined in \eqref{g1+}--\eqref{g+} are Carath\'eodory functions, which are convex and continuously differentiable in $\psi$ with
\begin{align*}
\nabla\fg^{1,\pm}(\psi)
&=
2\ell^\pm_{-} c^S\psi\pm2\ell^\pm_{-}b^S\pm2 c^{S\ell^\pm},\\
\nabla\fg^{2,\pm}(\psi)
&=
2\ell^\pm_{-}\int \big((1\pm\psi^\T u)^+u-u\big)F^{S}(du)\pm2\int(1\pm\psi^\T u)^+uyF^{S,\ell^\pm}(du,dy)\\
&\phantom{=\ }
\mp2\int (1\pm\psi^\T u)^-u(\ell^\mp_-+z)F^{S,\ell^\mp}(du,dz).
\end{align*}
\el
\bp
We only prove the assertion for $\fg^{2,-}$ as the arguments for the other functions are completely analogous or obvious. So we write $\fg^{2,-}$ as
\begin{align*}
\fg^{2,-}(\psi;S,\ell^+,\ell^-)
&=
\ell^-_-\int f_1(\psi,u)F^S(du)+\int f_2(\psi,u,y)F^{S, \ell^-}(du,dy)\\
&\phantom{=\ }
+\int \big(f_3(\psi,u)\ell^+_-+f_4(\psi,u,z)\big)F^{S,\ell^+}(du,dz)
\end{align*}
with
\begin{align*}
f_1(\psi,u)
&=
\big\{(1-\psi u)^+\big\}^2-1+2\psi u,\\
f_2(\psi,u,y)
&=
\big(\big\{(1-\psi u)^+\big\}^2-1\big)y,\\
f_3(\psi,u)
&=
\{(1-\psi u)^-\big\}^2,\\
f_4(\psi,u,z)
&=
\big\{(1-\psi u)^-\big\}^2z.
\end{align*}
Since $S\in\HzlocP$ and the jumps of $\ell^\pm$ are bounded by $1$, we obtain that $\int |u|^2F^S(du)$, $\int |u|^2|y|F^{S,\ell^-}(du,dy)$, $\int |u|^2|y|^2F^{S,\ell^-}(du,dy)$ and $\int |u|^2|z|F^{S,\ell^+}(du,dz)$ are finite. Combining this with the estimates
\begin{align*}
|f_1(\psi,u)|
&=
|\psi u|^2\one_{\{\psi u\leq 1\}}+|2\psi u-1|\one_{\{\psi u> 1\}}\leq 2|\psi|^2|u|^2,\\
|f_2(\psi,u,y)|
&=
\big|\big( (\psi u)^2-2\psi u \big) y\one_{\{\psi u\leq 1\}}-y\one_{\{\psi u>1\}}\big|\leq |\psi|^2|u|^2(|y|+|y|^2),\\
|f_3(\psi,u)|
&=
|\psi u-1|^2\one_{\{\psi u\leq 1\}}\leq |\psi|^2|u|^2,\\
|f_4(\psi,u,z)|
&=
|\psi u-1|^2|z|\one_{\{\psi u\leq 1\}}\leq |\psi|^2|u|^2|z|
\end{align*}
gives that $\fg^{2,-}$ is finite-valued for all $\psi\in\R^d$. The convexity of $\fg^{2,-}$ then follows immediately from the convexity of $f_1,\ldots,f_4$ in $\psi$. To verify the continuous differentiability of $\fg^{2,-}$, we want to differentiate under the integrals via an appeal to dominated convergence. To that end, we fix $\psi\in\R^d$, take an open ball $B_\ve(\psi)$ of radius $\ve>0$ around $\psi$ and estimate for $\xi\in B_\ve(\psi)$ the partial derivatives
\begin{align*}
|\nabla_{\psi}f_1(\xi,u)|
&=
|-2(1-\xi u)^+u+2u|\leq2|\xi uu|\one_{\{\xi u\leq 1\}}+2|u|\one_{\{\xi u> 1\}}\\
&\leq
2 (|\psi|+\ve)|u|^2+2|u|\one_{\left\{|u|>\frac{1}{|\psi|+\ve}\right\}}\leq 4(|\psi|+\ve)|u|^2=:h_1(u),\\
|\nabla_{\psi}f_2(\xi,u,y)|
&=
|-2(1-\xi u)^+uy|=2|\xi u||u||y|\one_{\{\xi u\leq 1\}}\leq2(|\psi|+\ve)|u|^2|y|=:h_2(u,y),\\
|\nabla_{\psi}f_3(\xi,u)|
&=
|2(1-\xi u)^-u|=2|1-\xi u| |u|\one_{\{\xi u\leq 1\}}\leq2(|\psi|+\ve)|u|^2=:h_3(u),\\
|\nabla_{\psi}f_4(\xi,u,z)|
&=
|2(1-\xi u)^-uz|=2|1-\xi u|\one_{\{\xi u\leq 1\}}|u||z|=:h_4(u,z).
\end{align*}
Since $h_1,\ldots,h_4$ are all integrable, we may indeed interchange differentiation and integration, and so $\fg^{2,-}$ is continuously differentiable in $\psi$. In particular, $\fg^{2,-}$ is continuous in $\psi$ and a Carath\'eodory function.
\ep
We next want to compute the drift of $J(\vt)$ for Theorem \ref{mainthm1}. Note below that the superscripts $\pm$ for $\ell$ only serve as indices; they do not denote positive and negative parts, unlike $V^\pm(x,\vt)$. While this notation may be slightly ambiguous, we found $\ell^{(\pm)}$ too heavy.
\bl\label{lcdc}
Let $\ell^\pm$ be $[0,1]$-valued semimartingales and set
$$j(\vt):=\big(V^+(x,\vt)\big)^2\ell^++\big(V^-(x,\vt)\big)^2\ell^-.
$$
For each $\vt\in\TBK$, we define the $K$-valued predictable process $\psi$ as in \eqref{defpsi} via
\be
\vt=:V_{-}^+(x,\vt)\psi+V_{-}^-(x,\vt)\psi+\one_{\{V_{-}(x,\vt)=0\}}\psi\label{defpsi2}.
\ee
Then $j^n(\vt):=\one_{D_n}\sint j(\vt)$ is a special semimartingale for each $D_n:=\{|\vt|\leq n\}\in\cP$ and $n\in\N$. In the canonical decomposition $j^n(\vt)=j^n_0(\vt)+M^{j^n(\vt)}+A^{j^n(\vt)}$, we have $A^{j^n(\vt)}=(\one_{D_n}\bar b^{\vt})\sint B$ with
\begin{align}
\bar b^{\vt}
&=
\big(V_{-}^+(x,\vt)\big)^2\big\{\fg^+(\psi;S,\ell^+,\ell^-)+b^{\ell^+}\big\}+\big(V_{-}^-(x,\vt)\big)^2\big\{\fg^-(\psi;S,\ell^+,\ell^-)+b^{\ell^-}\big\} \nonumber\\
&\phantom{=\ }
+\one_{\{V_{-}(x,\vt)=0\}}\Big(\int\big((\psi^\T u)^+\big)^2(\ell^+_{-}+y)F^{S,\ell^+}(du,dy)+\ell^-_-\psi^\T c^S\psi\nonumber\\
&\phantom{=\ }
+\int\big((\psi^\T u)^-\big)^2(\ell^-_{-}+z)F^{S,\ell^-}(du,dz)\Big).\label{defb2}
\end{align}
If $j(\vt)$ is special, then $b^{j(\vt)}=\bar b^\vt$.
\el
\bp 
The Meyer--It\^o formula (Theorem IV.71 in \cite{P04}) and integration by parts give
\begin{align*}
d\big(V^+(x,\vt)\big)^2
&=
2V_-^+(x,\vt)\vt \,dS+\one_{\{V_{-}(x,\vt)>0\}}\vt\, d[S^c]\vt
+\Delta\big(V^+(x,\vt)\big)^2-2V_-^+(x,\vt)\vt \Delta S,\\
d\big(V^-(x,\vt)\big)^2
&=
-2V_-^-(x,\vt)\vt \,dS+\one_{\{V_{-}(x,\vt)\leq0\}}\vt\, d[S^c]\vt
+\Delta\big(V^-(x,\vt)\big)^2+2V_-^-(x,\vt)\vt\Delta S
\end{align*}
and
\begin{align}
\one_{D_n}d\big\{ \ell^+\big(V^+(x,\vt)\big)^2 \big\}
&=
\one_{D_n}\big(V_-^+(x,\vt)\big)^2d\ell^++\one_{D_n}\ell^+_{-}\Big(2V_-^+(x,\vt)\vt \,dS \nonumber\\
&\phantom{=\ }
+\one_{\{V_{-}(x,\vt)>0\}}\vt\, d[S^c]\vt+\big\{\Delta\big(V^+(x,\vt)\big)^2-2V_-^+(x,\vt)\vt\Delta S\big\}\Big)\nonumber\\
&\phantom{=\ }
+2\one_{D_n}V_-^+(x,\vt)\vt\, d[S^c,(\ell^+)^c]+\one_{D_n}\Delta\big(V^+(x,\vt)\big)^2\Delta \ell^+ , \label{LV}\\
\one_{D_n}d\big\{ \ell^-\big(V^-(x,\vt)\big)^2 \big\}
&=
\one_{D_n}\big(V_-^-(x,\vt)\big)^2d\ell^-+\one_{D_n}\ell^-_{-}\Big(-2V_-^-(x,\vt)\vt \,dS \nonumber\\
&\phantom{=\ }
+\one_{\{V_{-}(x,\vt)\leq0\}}\vt\, d[S^c]\vt+\big\{\Delta\big(V^-(x,\vt)\big)^2+2V_-^-(x,\vt)\vt\Delta S\big\}\Big)\nonumber\\
&\phantom{=\ }
-2\one_{D_n}V_-^-(x,\vt)\vt\, d[S^c,(\ell^-)^c]+\one_{D_n}\Delta\big(V^-(x,\vt)\big)^2\Delta \ell^-.\label{LV2}
\end{align}
Since $\Delta V(x,\vt)=\vt \Delta S$, $S$ is in $\HzlocP$, $|\Delta\ell^\pm|\leq 1$ and $\vt$ is bounded on $D_n$, the supremum of the jumps of each term in \eqref{LV} and \eqref{LV2} is locally integrable. So Theorem III.36 in \cite{P04} implies that these terms are all special and we can calculate their compensators as
\begin{align*}
&\one_{D_n}\sint \big\{ \ell^+\big(V^+(x,\vt)\big)^2 \big\}\\
&\marteq
\one_{D_n}\big(V_-^+(x,\vt)\big)^2\sint A^{\ell^+}+(\one_{D_n}\ell^+_{-})\sint\big((2V_-^+(x,\vt)\vt)\sint A^S+\one_{\{V_{-}(x,\vt)>0\}}\sint[\vt\sint S^c]\big)\\
&\phantom{\marteq\ }
+\one_{D_n}\ell^+_{-}\big\{\big((V_{-}(x,\vt)+\vt u)^+\big)^2-\big(V_-^+(x,\vt)\big)^2-2V_-^+(x,\vt)\vt u\big\}\ast\nu^{S}\\
&\phantom{\marteq\ }
+\one_{D_n}\big\{\big(\big(V_{-}(x,\vt)+\vt u\big)^+\big)^2-\big(V_-^+(x,\vt)\big)^2\big\}y\ast\nu^{S,\ell^+}+2\one_{D_n}V_-^+(x,\vt)\sint [\vt\sint S^c,(\ell^+)^c],\\
&\one_{D_n}\sint \big\{ \ell^-\big(V^-(x,\vt)\big)^2 \big\}\\
&\marteq
\one_{D_n}\big(V_-^-(x,\vt)\big)^2\sint A^{\ell^-}+ (\one_{D_n}\ell^-_{-})\sint\big(-(2V_-^-(x,\vt)\vt)\sint A^S+\one_{\{V_{-}(x,\vt)\leq0\}}\sint[\vt\sint S^c]\big)\\
&\phantom{\marteq\ }
+\one_{D_n}\ell^-_{-}\big\{\big((V_{-}(x,\vt)+\vt u)^-\big)^2-\big(V_-^-(x,\vt)\big)^2+2V_-^-(x,\vt)\vt u\big\}\ast\nu^{S}\\
&\phantom{\marteq\ }
+\one_{D_n}\big\{\big((V_{-}(x,\vt)+\vt u)^-\big)^2-\big(V_-^-(x,\vt)\big)^2\big\}z\ast\nu^{S,\ell^-}-2\one_{D_n}V_-^-(x,\vt)\sint [\vt\sint S^c,(\ell^-)^c],
\end{align*}
where we denote by $\marteq$ equality up to a local martingale. Adding both equations and passing to differential characteristics gives
\begin{align}
A^{j^n(\vt)}
% (\one_{D_n}\bar b^{\vt})\sint B
&=
\one_{D_n}\Big(\one_{\{V_{-}(x,\vt)>0\}}\ell^+_{-}\vt c^S\vt+2V_-^+(x,\vt)\vt \big(\ell^+_{-}b^S+ c^{S,\ell^+}\big)+\big(V_-^+(x,\vt)\big)^2b^{\ell^+}\nonumber\\
&\phantom{=\ }
+\ell^+_{-}\int\big\{\big((V_{-}(x,\vt)+\vt u)^+\big)^2-\big(V_-^+(x,\vt)\big)^2-2V_-^+(x,\vt)\vt u\big\}F^S(du)\nonumber\\
&\phantom{=\ }
+\int\big\{\big((V_{-}(x,\vt)+\vt u)^+\big)^2-\big(V_-^+(x,\vt)\big)^2\big\}yF^{S,\ell^+}(du,dy)\nonumber\\
&\phantom{=\ }
+\one_{\{V_{-}(x,\vt)\leq0\}}\ell^-_{-}\vt c^S\vt-2V_-^-(x,\vt)\vt (\ell^-_{-}b^S+c^{S,\ell^-})+\big(V_-^-(x,\vt)\big)^2b^{\ell^-}\nonumber\\
&\phantom{=\ }
+\ell^-_{-}\int\big\{\big((V_{-}(x,\vt)+\vt u)^-\big)^2-\big(V_-^-(x,\vt)\big)^2+2V_-^-(x,\vt)\vt u\big\}F^{S}(du)\nonumber\\
&\phantom{=\ }
+\int\big\{\big((V_{-}(x,\vt)+\vt u)^-\big)^2-\big(V_-^-(x,\vt)\big)^2\big\}zF^{S,\ell^-}(du,dz)\Big)\sint B.\nonumber
\end{align}
By plugging in \eqref{defpsi2}, we obtain first
\begin{align*}
&\big((V_{-}(x,\vt)+\vt u)^\pm\big)^2-\big(V_-^\pm(x,\vt)\big)^2\\
&=
\big(V_-^\pm(x,\vt)\big)^2\big\{\big((1\pm\psi u)^+\big)^2-1\big\}
+\big(V_-^\mp(x,\vt)\big)^2\big((1\mp\psi u)^-\big)^2+\one_{\{V_{-}(x,\vt)=0\}}\big((\psi u)^\pm\big)^2
\end{align*}
and therefore also $A^{j^n(\vt)} = (\one_{D_n}\bar b^{\vt})\sint B$ with
\begin{align*}
\bar b^{\vt}
&=
\big(V_-^+(x,\vt)\big)^2\Big\{\ell^+_{-}\psi c^S\psi+2\psi(\ell^+_{-}b^S+c^{S\ell^+})+b^{\ell^+}\nonumber\\
&\phantom{=\ }
+\ell^+_{-}\int\big\{\big((1+\psi u)^+\big)^2-1-2\psi u\big\}F^S(du)+\int\big\{\big((1+\psi u)^+\big)^2-1\big\}yF^{S,\ell^+}(du,dy)\nonumber\\
&\phantom{=\ }
+\int\big((1+\psi u)^-\big)^2(\ell^-_{-}+z)F^{S,\ell^-}(du,dz)\Big\}\nonumber\\
&\phantom{=\ }
+\big(V_-^-(x,\vt)\big)^2\Big\{\ell^-_{-}\psi c^S\psi-2\psi(\ell^-_{-}b^S+c^{S\ell^-})+b^{\ell^-}\nonumber\\
&\phantom{=\ }
+\ell^-_{-}\int\big\{\big((1-\psi u)^+\big)^2-1+2\psi u\big\}F^{S}(du)+\int\big\{\big((1-\psi u)^+\big)^2-1\big\}zF^{S,\ell^-}(du,dz)\nonumber\\
&\phantom{=\ }
+\int\big((1-\psi u)^-\big)^2(\ell^+_{-}+y)F^{S,\ell^+}(du,dy)\Big\}\nonumber\\
&\phantom{=\ }
+\one_{\{V_{-}(x,\vt)=0\}}\Big\{ \int\big((\psi u)^+\big)^2(\ell^+_{-}+y)F^{S,\ell^+}(du,dy)\nonumber\\
&\phantom{=\ }
+\ell^-_-\psi c^S\psi+\int\big((\psi u)^-\big)^2(\ell^-_{-}+z)F^{S,\ell^-}(du,dz)\Big\}\nonumber
\end{align*}
after collecting terms. The assertion then follows by inserting the definitions of $\fg^\pm$.
\ep
%
%%%%%%%%%%%%%%%%%%%%%%%%%%%%%%%%%%%%%%%%%%%%%%%%%%%%%%%%
%
\section{Related work}
\label{sec:lit}
To round off the paper and put our contribution into perspective, we finally discuss the connections of our work to the existing literature. This naturally splits in two parts.
\subsection{The unconstrained case}\label{sec:rw:uc}
For (semimartingale) models without constraints, one key motivation to study the Marko\-witz problem has been the mean-variance hedging problem \eqref{MVH}. The solution of \eqref{MVH}, for an arbitrary payoff $H$, can be described more explicitly if one knows the variance-optimal martingale measure or the opportunity-neutral measure; see for example Theorem 4.6 in \cite{S01} and Theorem 4.10 in \cite{CK07}. Finding those measures is intimately linked to the approximation in $ \LiiP$ of the constant $1$ by stochastic integrals of $S$, i.e.~to \eqref{ap-1}. While there is a vast literature on mean-variance hedging, the most general results for these problems without constraints have been obtained by \v Cern\'y and Kallsen \cite{CK07}, and their work has also provided a lot of inspiration for our approach. We now quickly explain how the main results of \cite{CK07} can be obtained directly as special cases of our setting.

Suppose that there are no constraints so that $C\equiv K\equiv\RR^d$. The first key simplification is then that the opportunity processes $L^\pm$ agree so that we can write $L:=L^+=L^-$. One way to see this is to look at the proof of Proposition \ref{pJL} and note there that the distinction according to the sign of $x+\vt\sint S_\sigma$ becomes superfluous since $K$ is symmetric. Alternatively, one can look at the definitions of $\bL^\pm(\sigma)$ in \eqref{C} and observe that they agree for $+$ and $-$ because $\fK(0,\sigma;\sigma)$ contains with $\vp$ also $-\vp$. Again this only needs that $K$ is a cone and symmetric around $0$, but we shall exploit $K\equiv\RR^d$ later. Recall that $\TB = \overline{\Theta(\R^d)}$.

To get good properties for the (single) opportunity process $L$, we next suppose as in \cite{CK07} that there exists an equivalent $\sigma$-martingale measure (E$\sigma$MM) $Q$ for $S$ with $\frac{dQ}{dP}\in\LiiP$. (Because $S\in\HzlocP$, we then have that $\sup_{0\leq t\leq \tau_n}|S_t|\in L^1(Q)$ so that $Q$ is actually an equivalent local martingale measure (ELMM) for $S$.) Lemma \ref{lwac} then tells us that both $L$ and $L_-$ are strictly positive; this recovers Lemma 3.10 from \cite{CK07}. A substantial sharpening is given in Theorem \ref{T6.2} below.

Moving on to the local description in Section \ref{sec:ld}, we see from $L^+=L^-=L$ that we only need to consider a setting with $\ell^+=\ell^-=:\ell$. Then \eqref{g2+} reduces to
\begin{align*}
\mathfrak{g}^{2,+}(\psi)&=\ell_{-}\int\big((1+\psi^\T u)^2-1-2\psi^\T u\big)F^S(du)+\int\big((1+\psi^\T u)^2-1\big)yF^{S,\ell}(du,dy)\nonumber\\
&=\int(\psi^\T u)^2(\ell_-+y)F^{S,\ell}(du,dy)+\int2\psi^\T uyF^{S,\ell}(du,dy)\nonumber\\
&=\mathfrak{g}^{2,-}(-\psi),
\end{align*}
and therefore \eqref{g+} yields
\begin{align*}
\mathfrak{g}^{+}(\psi)&=\ell_{-}\psi^\T c^S\psi+2\ell_{-}\psi^\T b^S+2 \psi^\T c^{S\ell}+\mathfrak{g}^{2,+}(\psi)=\mathfrak{g}^{-}(-\psi).
\end{align*}
If in addition $\ell_-$ is strictly positive, we can rewrite this as
$$\mathfrak{g}^{+}(\psi)=\ell_{-}(\psi^\T\bar c\psi+2\psi^\T\bar b)=\mathfrak{g}^{-}(-\psi)$$
with
\begin{align}
\bar c&:=\bar c(S,\ell):=c^S+\int uu^\T\Big(1+\frac{y}{\ell_-}\Big)F^{S,\ell}(du,dy),\label{sec:rw:A}\\
\bar b&:=\bar b(S,\ell):=b^S+ \frac{c^{S\ell}}{\ell_-}+\int u\frac{y}{\ell_-}F^{S,\ell}(du,dy),
\label{sec:rw:B}
\end{align}
as in (3.25) and (3.23) in \cite{CK07}. So $\fg^\pm$ are quadratic functions and we can easily, by comple\-ting squares, find their minimisers and minimal values in explicit form. The result is
\be
\min_{\psi\in\RR^d}\mathfrak{g}^{+}(\psi)=\mathfrak{g}^{+}(\tpsi^+)=-\ell_-\bar b^\T (\bar c)^{-1}\bar b=\min_{\psi\in\RR^d}\mathfrak{g}^{-}(\psi)=\mathfrak{g}^{-}(\tpsi^-)\label{sec:rw:C}
\ee
with
\begin{equation}
\tpsi^+=-\tpsi^-=:\tpsi=-(\bar c)^{-1}\bar b=:-\bar a,
\label{sec:rw:C1}
\end{equation}
where $(\bar c)^{-1}$ denotes the Moore--Penrose pseudoinverse of $\bar c$. We remark that this is well defined whenever a minimiser exists, hence in particular if there is an optimal strategy.

Under the assumption (made in \cite{CK07}) that there is an E$\sigma$MM $Q$ for $S$ with \hbox{$\frac{dQ}{dP}\in\LiiP$}, Theorem 2.16 for $C\equiv\RR^d$ tells us that $G_T(\TB)$ is closed in $\LiiP$. The same is true for
$$
G_T(\TB\mathbbm{1}_{\rrbracket \tau,T\rrbracket})=G_T(\overline{\Theta(\RR^d\mathbbm{1}_{\rrbracket \tau,T\rrbracket})})
$$
for any stopping time $\tau$, and so \eqref{A} has a solution $\tvp^{(x,\tau)}$ for every pair $(x,\tau)$. Corollary \ref{corld} thus allows us to identify $\tvp^{(x,\tau)}$; indeed, $\tpsi^+ = -\tpsi^- = \tpsi$ reduces the SDE \eqref{eq:corld:b} to
$$
d V_t^{(x,\tau)}=V_{t-}^{(x,\tau)}\widetilde\psi_t\one_{\rrbracket \tau,T\rrbracket }\,dS_t,\quad V_0^{(x,\tau)}=V_\tau^{(x,\tau)}=x
$$
whose solution is of course
$$
V^{(x,\tau)}=x\,\E\big((\tpsi\one_{\rrbracket \tau,T\rrbracket })\sint S\big)=x\,\E\big((-\bar a\one_{\rrbracket \tau,T\rrbracket })\sint S\big),
$$
and so \eqref{eq:corld:c} yields
\be
\tvp^{(x,\tau)} =V_{-}^{(x,\tau)}\widetilde\psi\one_{\rrbracket \tau,T\rrbracket }=-x\,\E\big((-\bar a\one_{\rrbracket \tau,T\rrbracket })\sint S\big)_- \bar a\one_{\rrbracket \tau,T\rrbracket }.\label{eq:corld:CK}
\ee
This recovers Lemma 3.7 from \cite{CK07}.

One major simplification in the unconstrained case is that we no longer need to distinguish between the cases $V_-(x,\tvp)>0$ and $V_-(x,\tvp)<0$ because there is only one opportunity process $L$. In terms of the discussion before Corollary \ref{corld}, we no longer need to worry about jumps of $V(x,\tvp)$ across $0$ since these do not affect the description of $L$. All we need is to be able to ``restart $V(x, \tvp)$ when it jumps to $0$'', which is the important insight obtained by \v Cern\'y and Kallsen \cite{CK07}. The adjustment process $\widetilde a$ from \cite{CK07} is moreover seen to be given by $\widetilde a=\bar a=-(\bar c)^{-1}\bar b=-\tpsi$, by comparing \eqref{eq:corld:CK} to (3.12) in \cite{CK07}.

The above result highlights an important difference between our approach and that in \cite{CK07}. We obtain our results by systematically using stochastic control ideas and in particular the martingale optimality principle (MOP). To illustrate this with an example, we see from the above that $\widetilde a=-\tpsi$ is obtained as the minimiser of the function $\fg$, which means that we exploit the MOP by using that the drift of $J(\vt)$ must vanish for the optimal strategy. In contrast, \v Cern\'y and Kallsen \cite{CK07} obtain $\widetilde a$ by closely examining the structure of the optimal strategies $\tvp^{(x,\tau)}$ for variable $\tau$, and they prove its properties using the optimality of $\tvp^{(x,\tau)}$ via martingale orthogonality conditions. They do not explicitly use dynamic programming and never mention the MOP.

The next proposition summarises the most important results for the unconstrained case $C\equiv\RR^d$. We give no proof; this all follows directly by specialising our earlier results. 

\begin{prop}\label{prop:sum}
Suppose that $S$ is in $\HzlocP$. Then:
\bi
\item[\bf{1)}] There exists an RCLL submartingale $L=(L_t)_{0\leq t\leq T}$, called \emph{opportunity process}, such that for each $x\in\RR$ and $\tau\in\cS_{0,T}$, the process
$$
\textstyle
J_t(\vt;x,\tau)=\big( x+\int_\tau^t\vt_udS_u \big)^2L_t,\quad 0\leq t\leq T
% \label{E2}
$$
is a submartingale for every $\vt\in\TB$ with $\vt=0$ on $\llbracket0,\tau\rrbracket$. Moreover, $J(\tvt;x,\tau)$ is a martingale for $\tvt\in\TB$ with $\tvt=0$ on $\llbracket 0,\tau\rrbracket$ if and only if $\tvt=\tvp^{(x,\tau)}$ is optimal for \eqref{A}. The process $L$ is given explicitly as an RCLL version of
$$
\textstyle
\bL(t):=\essinf\big\{ E\big[|1-\int_t^T \vp_u \,dS_u|^2\big|\F_t\big] \,\big|\, \text{$\vp\in\TB$ with $\vp=0$ on $\llbracket0,t\rrbracket$}\big\},\quad 0\leq t\leq T.
$$
\item[\bf{2)}] Suppose that $L$ and $L_-$ are both $>0$ and that there exists a solution $\tvp^{(1,\tau)}$ to \eqref{A} with $x=1$ for any stopping time $\tau$. Then the joint differential characteristics of $(S,L)$ satisfy
\be
\text{$b^{L}=L_-\bar b^\T (\bar c)^{-1}\bar b$}\label{eq:corld:aCK}
\ee
and we have $V(1,\tvp^{(1,\tau)})=\E\big((-\bar a\one_{\rrbracket \tau,T\rrbracket })\sint S\big)$ with $\bar a=(\bar c)^{-1}\bar b$. A sufficient condition for the assumptions in 2) is that there exists an E$\sigma$MM $Q$ for $S$ with $\frac{dQ}{dP}\in\LiiP$.
\item[\bf{3)}] Conversely, let $\ell$ be a semimartingale such that
\bi
\item[\bf{a)}] $\ell$ and its left limit $\ell_{-}$ are $(0,1]$-valued and $\ell_T=1$.
\item[\bf{b)}] The joint differential characteristics of $(S,\ell)$ satisfy
$$
b^{\ell}=\ell_-\bar b^\T (\bar c)^{-1}\bar b.
% \label{eq:jc}
$$
\item[\bf{c)}] For $\bar a:=(\bar c)^{-1}\bar b$, we have that
$$
\bar\lambda^{(\tau)}:= \E\big( (-\bar a\mathbbm{1}_{\rrbracket \tau,T\rrbracket})\sint S\big)_- \bar a \mathbbm{1}_{\rrbracket \tau,T\rrbracket})\in\TB .
$$
\ei
Then $\tvp^{(1,\tau)}:=-\bar\lambda^{(\tau)}$ is the solution to \eqref{ap} with $x=1$ for each $\tau\in\cS_{0,T}$, and $L:=\ell$ is the opportunity process.
\ei
\end{prop}

Note that the equation \eqref{eq:corld:aCK} for the joint differential characteristics of $(S,\ell)$ is the same as (3.32) in \cite{CK07}. Moreover, parts 2) and 3) of Proposition \ref{prop:sum} essentially recover Theorem 3.25 of \cite{CK07}; our result is actually even stronger since we do not need the assumption from \cite{CK07} that $\E( (-\bar a\mathbbm{1}_{\rrbracket \tau,T\rrbracket})\sint S)\ell$ is of class (D) for each stopping time $\tau\in\cS_{0,T}$.

The results of \v Cern\'y and Kallsen \cite{CK07} show (as repeated in part 2) of Proposition \ref{prop:sum}) that a sufficient condition for the existence of all optimal strategies $\tvp^{(1,\tau)}$ for $\tau\in\cS_{0,T}$ as well as for strict positivity of $L$ and $L_-$ is the existence of an E$\sigma$MM $Q$ for $S$ with $\frac{dQ}{dP}\in\LiiP$. Our next theorem sharpens this into a precise characterisation by giving \emph{necessary and sufficient} conditions. This result is also one reason why we have introduced the notion of $(\E,\ZN)$-martingales in the precise form of Section \ref{sec:pf}.
\begin{thm}
\label{T6.2}
For $S\in\HzlocP$, the following are equivalent:
\bi
\item[1)] The opportunity process $L$ and its left limit $L_{-}$ are $(0,1]$-valued and there exists a solution $\tvp^{(1,\tau)}$ to \eqref{A} with $x=1$ for any stopping time $\tau\in\cS_{0,T}$.
\item[2)] There exist $N\in\MznlocP$ and $\ZN$ such that $(\E,\ZN)$ with $\E=\E(N)$ is regular and square-integrable and $S=S_0+M-\la M, N\ra$ is an $\E$-local martingale.
\ei
\end{thm}
\bp
The implication ``$2)\Longrightarrow 1)$'' is easy. Indeed, the closedness in $\LiiP$ of $G_T(\TBC)$ obtained from Theorem \ref{thmclosedTBC} implies the existence of all the $\tvp^{(1,\tau)}$ by taking $C = \RR^d \one_{\rrbracket\tau,T\rrbracket}$, and strict positivity of $L$ and $L_-$ is from Lemma \ref{lwac}. We prove the converse implication ``$1)\Longrightarrow 2)$'' in several steps.

1) Fix $\tau$ and use Lemma \ref{lpfw} to write 
$
V(1,\tvp^{(1,\tau)}) = \cE(\tpsi^{(1,\tau)}\sint S) = \cE( (\tpsi^{(1,\tau)}\one_{\rrbracket\tau,T\rrbracket})\sint S)
$.
As in Lemma \ref{mart}, using that $L^+ = L^- = L$, consider the process $\widetilde M^{(1,\tau)} = V(1,\tvp^{(1,\tau)}) L$ and the square-integrable martingale 
$
\one_{\rrbracket\tau,T\rrbracket}\sint\widetilde M^{(1,\tau)} = \one_{\rrbracket\tau,T\rrbracket}\sint (V(1,\tvp^{(1,\tau)}) L)
$.
Because $L_->0$, we can write $L = L_0\,\cE(K')$. Moreover, Corollary \ref{corld} and its proof give that $\tpsi^{(1,\tau)}$ coincides on the set
$
\rrbracket\tau,T\rrbracket \cap \{V_-(1,\tvp^{(1,\tau)})\ne0\}
$
with the minimiser $\tpsi$ of the function $\fg$, which is $\tpsi = -\bar a = -(\bar c)^{-1}\bar b$ by \eqref{sec:rw:C1}, so that
$
V(1,\tvp^{(1,\tau)}) = \cE((-\bar a \one_{\rrbracket\tau,T\rrbracket})\sint S)
$.
This implies
\begin{align}
\widetilde M^{(1,\tau)}
\nonumber
&=
L^\tau + \one_{\rrbracket\tau,T\rrbracket}\sint \big( V(1,\tvp^{(1,\tau)}) L_0\,\cE(K') \big)
\\
&=
L^\tau + \one_{\rrbracket\tau,T\rrbracket}\sint \Big( \cE\big( (-\bar a\one_{\rrbracket\tau,T\rrbracket})\sint S \big) L_\tau \, \cE\big( \one_{\rrbracket\tau,T\rrbracket}\sint K' \big) \Big)
\nonumber
\\
&=
L^\tau \cE\big( \one_{\rrbracket\tau,T\rrbracket}\sint N \big)
\label{Mteq}
\end{align}
by Yor's formula, with $N := -\bar a\sint S + K' - [\bar a\sint S, K']$. Moreover, by Lemma \ref{mart} for 
$\vt := \pm \one_{\rrbracket\tau_n,\tau_{n+k}\rrbracket}$ for a localising sequence with $S^{\tau_m}\in\HzP$ for all $m$, we obtain that the product of ${}^{\tau_n} S$ and $\widetilde M^{(1,\tau)}$ is for each $n$ a local martingale (with $(\tau_{n+k})_{k\in{\mathbb N}}$ as localising sequence).

2) At the end of step 1), we have glossed over a point that we must settle now. While \eqref{Mteq} is correct as it stands, the subsequent definition of $N$ on all of $\llbracket0,T\rrbracket$ requires us to show that $\bar a$ is in $\cL(S)$. To do that, we recall that $K' = \frac{1}{L_-}\sint L$ (this is called the extended mean-variance tradeoff process in Definition 3.11 in \cite{CK07}) and introduce the opportunity-neutral measure $P^*\approx P$ by
$
\frac{dP^*}{dP} := \frac{L_T}{E[L_0] \, \cE(A^{K'})_T}
$.
Then Girsanov's theorem (see Lemma A.9 in \cite{CK07}) gives as in the proof of Lemma 3.17 in \cite{CK07} that
$
b^{S,P^*} = \frac{\bar b}{1+\Delta A^{K'}}
$
and
$
[S]^{{\bf p},P^*} = \widetilde c^{S,P^*}\sint B = \frac{\bar c}{1+\Delta A^{K'}}\sint B
$.
Note that $A^{K'}$ is increasing because $L$ is a submartingale, and Corollary \ref{corld} with \eqref{sec:rw:C} gives
$$
% \begin{equation}
\textstyle
A^{K'}
=
\frac{1}{L_-}\sint A^L
=
\frac{b^L}{L_-}\sint B
=
\big( - \frac{1}{L_-} \min_{\psi\in\RR^d} \fg(\psi;L) \big)\sint B
=
(\bar b^\top (\bar c)^{-1} \bar b) \sint B.
% \label{AKsexpr}
% \end{equation}
$$
So we obtain from $\bar a = - (\bar c)^{-1} \bar b$ and since $[S]^{{\bf p}, P^*} - \langle M^{S,P^*} \rangle$ is nonnegative definite that
$$
\textstyle
\int |\bar a \, dA^{S,P^*}| + \int \bar a^\top d\langle M^{S,P^*}\rangle \bar a
=
( |\bar a^\top b^{S,P^*}| + \bar a^\top \widetilde c^{M,P^*} \bar a)\sint B
\le
2 \frac{\bar b^\top (\bar c)^{-1} \bar b}{1+\Delta A^{K'}} \sint B
\le
2 A^{K'},
$$
which shows that $\bar a$ is in both $\cL(A^{S,P^*})$ and $\cL^2_{\rm loc}(M^{S,P^*})$ and therefore in $\cL(S)$. Hence $N$ is well defined and a semimartingale. As in Section \ref{sec:pf}, define the stopping times $T_0:=0$ and
$T_{m+1}=\Inf\{t>T_m\, |\, {}^{T_m}\E(N)_t=0\}\wedge T$, and note that $(T_m)$ increases to $T$ stationarily.

3) Step 1) with $\tau=T_m$ implies that 
$
\one_{\rrbracket T_m,T\rrbracket}\sint \widetilde M^{(1,T_m)} = L_{T_m} \one_{\rrbracket T_m,T\rrbracket} \sint \cE( \one_{\rrbracket T_m,T\rrbracket} \sint N)
$
is for each $m$ a square-integrable martingale. By Remark \ref{Nrem}, this implies that $N$ is in $\MznlocP$ because $L>0$.  Then step 1) also shows that $(\E,\ZN)$ with $\E = \E(N)$ and $\ZN = L$ is regular and square-integrable, since the product of $L^{T_m}$ and ${}^{T_m}\E(N) = \E( \one_{\rrbracket T_m,T\rrbracket}\sint N)$ is $\widetilde M^{(1,T_m)}$. Finally, step 1) with $\tau_n$ replaced by $\tau_n\wedge T_m$ yields for $n\to\infty$ that $S$ is an $\E$-local martingale. This ends the proof.
\ep
An alternative description of $L$ and hence of the optimal strategies is via the BSDE \eqref{L-BSDE} in Corollary \ref{cor:BSDE}. Combining \eqref{sec:rw:C} with the fact that $\fh^\pm=\fg^\pm$ in Section \ref{sec:ld}, we obtain that the BSDE system \eqref{L-BSDE} (for $L^\pm$) collapses to the single BSDE (for $L$)
$$
L = (L_-\bar b^\T (\bar c)^{-1}\bar b)\sint B + H^{L}\sint S^c + W^{L}\ast(\mu^S-\nu^S) + N^{L}, \quad L_T=1.
$$
By also using \eqref{sec:rw:A}, \eqref{sec:rw:B} and \eqref{delta}--\eqref{beta}, we can rewrite the drift term (with respect to $B$) into a more explicit form and obtain 
\begin{align}
L&=H^{L}\sint S^c+W^{L}\ast(\mu^S-\nu^S)+N^{L}
\nonumber\\
&\phantom{=\ }+\bigg\{\bigg(b^S+ c^S\frac{H^L}{L_-}+\int \frac{\Delta A^L+W^L(u)-\WhL}{L_-}uF^S(du)\bigg)^\T\nonumber\\
&\phantom{=\ }\times \bigg(c^S+\int uu^\T\Big(1+\frac{\Delta A^L+W^L(u)-\WhL}{L_-}\Big)F^S(du)\bigg)^{-1}\nonumber\\
&\phantom{=\ }\times\bigg(b^S+ c^S\frac{H^L}{L_-}+\int \frac{\Delta A^L+W^L(u)-\WhL}{L_-}uF^S(du)\bigg)L_-\bigg\}\sint B,
\quad
L_T=1.
\label{gamma}
\end{align}
This is much simpler than the constrained case because we no longer have a coupled system of BSDEs (for $L^\pm$). Note that \eqref{gamma} has one more term than the otherwise identical equation (3.37) in \cite{CK07}; it seems that \v Cern\'y and Kallsen \cite{CK07} have somewhere lost $\Delta A^L$, as has also been noted by other authors.
\subsection{The continuous case}
To the best of our knowledge, all results on the Markowitz problem under constraints in continuous-time models have been obtained when $S$ is \emph{continuous}. Before discussing individual papers, we therefore explain how our results simplify for continuous $S$.

First of all, Lemma \ref{lpfw} yields that $V(x,\tvp^{(x,\tau)})=x\,\E(\tpsi^{(x,\tau)}\sint S)$. So if \eqref{ap} (when we start from $\tau=0$) has a solution, the process $V(x,\tvp^{(x,0)})$ has a \emph{unique sign} on all of $\llbracket 0,T\rrbracket$ because the stochastic exponential of a continuous process never hits $0$. One can then show with some extra work that
$$
\bar\vp^{(x,\tau)} := x\,\E\big((\tpsi^{(x,0)}\mathbbm{1}_{\rrbracket \tau,T\rrbracket})\sint S\big)\tpsi^{(x,0)}\mathbbm{1}_{\rrbracket \tau,T\rrbracket}
=
\frac{x}{V_\tau(x,\tvp^{(x,0)})}\tvp^{(x,0)}\mathbbm{1}_{\rrbracket \tau,T\rrbracket}
$$
is optimal for \eqref{A} (when we start from $\tau$); more precisely, this can be done if we have the existence of an optimal strategy $\tvp^{(x,\tau)}$ for all $(x,\tau)$ or if the constraints cor\-re\-spon\-dence $C$ has convex closed cones as values.  So if $S$ is continuous, we basically do not need to study all the conditional problems; it is enough to understand and describe $\tvp^{(x,0)}$.

In the local description in Section \ref{sec:ld}, we next see in \eqref{g2+} that $\fg^{2,\pm}\equiv0$ when $S$ has no jumps; so \eqref{g+} gives $\fg^{\pm}=\fg^{1,\pm}$ and \eqref{g1+} shows that $\fg^+$ and $\fg^-$ only depend on $\ell^+$ and $\ell^-$, respectively. This implies in turn that the two coupled equations in \eqref{mainthm1:eq:1} in Theorem \ref{mainthm1} \emph{decouple}; and since we have already seen above that $V(x,\tvp)$ has a unique sign on $\llbracket 0,T\rrbracket$, we need in fact only one of those two equations \mbox{(depending on the sign of $x$).}

To describe the optimal strategy $\tvp^{(x,0)}$, we must find the minimiser $\tpsi^{(x,0)}$ of $\fg^+$ or $\fg^-$ (depending on the sign of $x$). Because $\fg^\pm$ are simple quadratic functions of $\psi$, as the terms $\fg^{2,\pm}$ are absent, finding their minimisers over the constraint set $K$ is straightforward in principle. But explicit (closed form) expressions can be expected only in special cases.

Conversely, Theorem \ref{vt1} allows us to construct a solution $\tvp^{(x,0)}$ to \eqref{ap} from a solution to the BSDEs in \eqref{BSDE2}. Those equations take the more explicit form
\be
\ell^\pm=-\Inf_{\psi\in K}\fh^\pm(\psi;S,\ell^\pm)\sint B + H^{\ell^\pm}\sint M + N^{\ell^\pm},
\quad \ell^\pm_T=1
\label{E-2}
\ee
with
$$
\mathfrak{h}^\pm(\psi;S,\ell^\pm)=\ell^\pm_{-}\psi^\T c^S\psi\pm2\ell^\pm_{-}\psi^\T b^S\pm2 \psi^\T c^{S}H^{\ell^\pm} .
$$
In the unconstrained case $C\equiv K\equiv\R^d$, we can find the minimal value of $\fh^\pm$ explicitly by completing the square. Since we then also need not distinguish between $\ell^+$ and $\ell^-$, as seen in Section \ref{sec:rw:uc}, the BSDE \eqref{E-2} becomes (after doing the computations)
\begin{equation}
L = {}H^{L}\sint M+N^{L} + \bigg\{\bigg(b^S+c^S\frac{H^L}{L_-}\bigg)(c^S)^{-1}\bigg(b^S+c^S\frac{H^L}{L_-}\bigg)L_-\bigg\}\sint B,\quad
L_T=1.
\label{F-2}
\end{equation}
This equation can also be found in Kohlmann and Tang \cite{KT02}, Mania and Tevzadze \cite{MT03} or Bobrovnytska and Schweizer \cite{BS04}, among others. Of course, \eqref{F-2} can also be obtained as a special case of \eqref{gamma} by simply dropping there all the jump terms. Note that even if $S$ is continuous, $L$ need not be, due to the presence of the orthogonal martingale term $N^L$.

After these general remarks, let us now discuss and compare the most important results in the literature so far.

We start with Hu and Zhou \cite{HZ04}, Labb\'e and Heunis \cite{LH07} and Li, Zhou and Lim \cite{LLZ02}. They all use for $S$ a multidimensional It\^o process model as in Example \ref{Exa:Ito} of the form
\be
dS_t=\mathrm{diag}(S_t)\big((\mu_t-r_t{\bf 1})\, dt+\sigma_t\, dW_t\big)\label{G}
\ee
with a vector drift process $\mu$ and a matrix volatility process $\sigma$. An important assumption is that $\dim S= \dim W$ and that $\sigma$ is invertible (even uniformly elliptic); this means that the model without constraints is complete and implies that the projection $\Pi^S$ on the predictable range of $S$ is simply the identity. Finally, the constraints are given by closed convex cones $K$ which are constant (i.e.~do not depend on $t$ or $\om$).

In \cite{HZ04}, the approach is to first study a more general constrained stochastic linear-quadratic (LQ) control problem and then derive results for the Markowitz problem as a special case. One inherent disadvantage is that this usually provides less intuition and insight than a direct approach as in our paper. At the more abstract level, \cite{HZ04} prove verification theorems; they show how solutions to certain BSDEs induce solutions to certain LQ control problems and also prove existence of solutions to their BSDEs under suitable conditions. In the context of the model \eqref{G}, one key assumption is that the instantaneous \emph{Sharpe ratio} process 
$
\bar\lambda := \sigma^{-1}(\mu-r{\bf 1}) = \sigma^\top (\sigma\sigma^\top)^{-1}(\mu-r{\bf 1})
$ 
is uniformly \emph{bounded}; this is exploited to prove solvability of the BSDEs by using results of Kobylanski \cite{K00}. Moreover, the arguments exploit (via the use of BSDE comparison theorems) that the opportunity processes $L^\pm$ are continuous since the filtration generated by the driving Brownian motion has no discontinuous martingales. Boundedness of $\bar \lambda$ also implies the existence of an E$\sigma$MM $Q$ for $S$ with $\frac{dQ}{dP}\in\LiiP$; in fact, one can take for $Q$ the minimal martingale measure given by $dQ=\E(-\bar\lambda \sint W)_T \, dP$. Theorem \ref{thmclosedTBC} then implies the closedness in $ \LiiP$ of $G_T(\TBK)$ and hence the solvability of \eqref{ap}. Actually, boundedness of $\bar\lambda$ even implies that the minimal martingale measure satisfies the reverse H\"older inequality $R_2(P)$, so that $G_T(\TSK)$ is closed in $\LiiP$; see Remark \ref{compare}. Moreover, the opportunity process $L$ from Proposition \ref{prop:sum} is uniformly bounded away from 0 due to the reverse H\"older inequality $R_2(P)$, and hence so are both $L^\pm$ because they dominate $L$. As already commented before Lemma \ref{maximal}, the solution to \eqref{BSDE2} is then also unique within that class of processes. But for applications, one serious drawback of assuming $\bar\lambda$ bounded is that this restrictive condition is often hard to check or even not satisfied in specific (e.g.~Markovian) models for $S$. Moreover, we could not find in \cite{HZ04} any explanation where the BSDEs come from so that the presentation seems to us not fully transparent. One simple illustration is that the authors of \cite{HZ04} also observe that one needs only one of the two BSDEs; but their explanation seems to miss that this is directly due to the continuity of $S$, as explained above before \eqref{E-2}.

In \cite{LH07}, the final setting is even more special since the coefficients $\mu, r, \sigma$ in \eqref{G} are all deterministic functions. Labb\'e and Heunis \cite{LH07} use convex duality to obtain existence and the structure of the solution to the Markowitz problem, by first solving a dual problem and then constructing from that the desired primal solution. More precisely, existence is proved for random coefficients and even (fixed) convex closed, but not necessarily conic, constraints if $\bar\lambda=\sigma^{-1}(\mu-r{\bf 1})$ is bounded (as in \cite{HZ04}). However, the results on the \emph{structure} of the optimal portfolio are obtained by first studying and solving the HJB equation for the dual problem, and this hinges crucially on the assumption of deterministic coefficients. It also needs closed convex cones for the constraints. 
From our perspective, the use of duality is in general not really necessary to obtain the structure of the solution to the primal problem. Duality is very often useful for proving the existence of a (primal) solution; but if that is achieved differently (or assumed), structural results about the solution can usually be derived directly in the primal setting, as we have done here.

Finally, one of the earliest papers on the Markowitz problem under constraints in a continuous-time setting is due to Li, Zhou and Lim \cite{LLZ02}. The coefficients $\mu, r, \sigma$ there are deterministic functions, $\bar\lambda=\sigma^{-1}(\mu-r{\bf 1})$ is again bounded, and constraints are given by $C\equiv K\equiv\RR^d_+$ (no shortselling). The treatment in \cite{LLZ02} combines LQ control with Markovian and PDE techniques; instead of working with BSDEs as in \cite{HZ04}, the authors of \cite{LLZ02} study the (primal) HJB equation associated to the Markowitz problem, construct for that a viscosity solution, and use a verification result to then derive the optimal strategy. A major step in their proof is to deal with a potential irregularity in the HJB equation (the set $\Gamma_3$ in \cite{LLZ02}, where $v(t,x)=0$).
From our general perspective, there are two comments. One is that a (well-hidden) assumption in \cite{LLZ02} is that the vector $\mu-r{\bf 1}$ is in $\RR^d_+$ (since the coefficient $B$ in the abstract problem (3.1) in \cite{LLZ02} must lie in the positive orthant). By looking at our functions $\fg^\pm=\fg^{1,\pm}$ in \eqref{g1+} and using that $K\equiv\RR^d_+$, we then directly obtain as minimisers $\tpsi^+=0$ and $\tpsi^- = (\sigma^\T)^{-1} \bar\pi$, where $\bar\pi$ denotes the projection on $\sigma^\T K = \sigma^\T \RR^d_+$ of $\bar\lambda = \sigma^{-1} (\mu - r{\bf 1})$; so the optimal strategy is almost directly given. Secondly, the fact that $V(x,\tvp)$ has a unique sign implies that the potential irregularity in the HJB equation is actually not relevant since the optimiser will not go there; this explains why there is no genuine smoothness problem in \cite{LLZ02}.

While all the above papers consider models which are complete without constraints, there has also been some recent work going beyond such restrictive setups; we mention here Jin and Zhou \cite{JZ07} and Donnelly \cite{D08}. Both use duality techniques to prove the existence of a solution; \cite{D08} has an It\^o process model with regime-switching coefficients and (deterministic and constant) convex constraints, while \cite{JZ07} studies no-shortselling constraints ($C\equiv K\equiv \RR^d_+$) in an incomplete It\^o process model. The latter paper also obtains the optimal strategy more explicitly for the special case of deterministic parameters $\mu, r, \sigma$; this is possible because (like in \cite{LH07}) the dual problem becomes much simpler under that condition. All in all, it seems fair to say that even for continuous $S$, our results on the \emph{structure} of the optimal strategy in the Markowitz problem under constraints contain and substantially extend all the available literature so far.

The last statement needs an important clarification. We focus here on constraints on \emph{strategies} and there in particular on the \emph{structure of the optimiser} for the Markowitz problem. There have been quite a few papers on the Markowitz problem (usually in the form \eqref{CMP} of minimising the variance  subject to a given mean for the final wealth) with the additional constraint of having a \emph{nonnegative wealth} process. One of the earliest papers on this topic is due to Korn and Trautmann \cite{KT95}, and more recent contributions include Bielecki, Jin, Pliska and Zhou \cite{BJPZ} and Xia \cite{X05}. In most cases, the discussion and solution goes as follows. If one has a good equivalent martingale measure $Q$, say, then nonnegative wealth $V(x,\vt)\geq 0$ as a process is equivalent to having nonnegative final wealth, $V_T(x,\vt)\ge0$. If one also has a complete model, every final payoff is replicable and so it is enough to solve the static Markowitz problem over (nonnegative) final wealth only. This is done in \cite{KT95} via duality and utility-based techniques and in \cite{BJPZ} via Lagrange multipliers. The paper by Xia \cite{X05} is a little different; it actually reduces the problem of minimising $E[|y-V_T(x,\vt)|^2]$ for continuous $S$ and $y>x$ by observing (and proving) that it is optimal to first minimise the expected squared shortfall $E[|(y-V_T(x,\vt))^+|^2]$ and then stop the corresponding wealth process as soon as it hits $y$. But in all these cases, a nonnegative wealth constraint is substantially easier to deal with than constraints imposed on strategies.
\bigskip
\begin{ak}
Financial support by the National Centre of Competence in Research ``Financial Valuation and Risk Management'' (NCCR FINRISK), Project D1 (Mathematical Methods in Financial Risk Management) is gratefully acknowledged. The NCCR FINRISK is a research instrument of the Swiss National Science Foundation.
\end{ak}

\end{document}